\documentclass[11pt]{article}

\usepackage[text={6.8in,8.5in},centering]{geometry}
\usepackage{times}
\usepackage{amsfonts}
\usepackage{amsmath}
\usepackage{systeme}

\usepackage{amssymb}
\usepackage{amsthm}
\usepackage{graphicx}
\usepackage{url}
\usepackage{hyperref}
\usepackage[table]{xcolor}

\usepackage{color}
\usepackage{tcolorbox}
\usepackage[]{mdframed}
\usepackage{subcaption}
\usepackage{bbm}

\usepackage{mdframed}

\usepackage{float}

\graphicspath{{LimitFig/}}  

\usepackage{cancel}

\newtheorem{lemma}{Lemma}[section]
\newtheorem{theorem}[lemma]{Theorem}
 \newtheorem{definition}[lemma]{Definition}
\newtheorem{proposition}[lemma]{Proposition}
\newtheorem{corollary}[lemma]{Corollary }

\newtheorem{remark}[lemma]{Remark}

\usepackage{tikz}
\usetikzlibrary{trees}

\newcommand{\bra}[1]{\ensuremath{\langle#1|}}
\newcommand{\ket}[1]{\ensuremath{|#1\rangle}}

\newcommand{\mc}[1]{\ensuremath{\mathcal{#1}}}   
\newcommand{\mb}[1]{\ensuremath{\mathbbm{#1}}}   
\newcommand{\ms}[1]{\ensuremath{\mathsf{#1}}} 
\newcommand{\mt}[1]{\ensuremath{\mathtt{#1}}}
\newcommand{\abbr}[1]{\ensuremath{\mathtt{#1}}}  

\renewcommand{\bar}{\overline}

\newcommand{\Hgen}{\widetilde{\mb{H}}}
\newcommand{\Sgen}{\widetilde{\mb{S}}}
\newcommand{\Egen}{\widetilde{E}}

\newcommand{\DI}{{\Delta \mb{I}}}
\newcommand{\YH}{{\Delta \mb{H}}}

\newcommand{\co}{{\ms{co}}}

\newcommand{\X}{\ms{X}}
\newcommand{\Z}{\ms{Z}}
\newcommand{\sZ}{\ms{\tilde{Z}}}
\newcommand{\XX}{\ms{XX}}

\newcommand{\Sop}[1]{\mathrm{S}_{#1}}  

\numberwithin{equation}{section}

\newcommand{\ver}{{\ms{V}}}
\newcommand{\edge}{{\ms{E}}}

\newcommand{\Jxx}{J_{\ms{xx}}}
\newcommand{\Jzz}{J_{\ms{zz}}}

\newcommand{\Hcore}{\mb{H}_{\ms{core}}}

\newcommand{\Heff}{\mb{H}_1^{\ms{eff}}}

\newcommand{\mdef}{\stackrel{\mathrm{def}}{=}}

\newcommand{\cl}{\ms{Clique}}

\usepackage{mathtools}

\usepackage{tabularx}
\usepackage{booktabs}
\usepackage{braket}
\usepackage{color}
\usepackage{tcolorbox}
\usepackage{array, xcolor, colortbl}
\newcolumntype{L}{>{\raggedright\arraybackslash}X}

\usepackage{enumitem}

\newcommand{\shz}[1]{\tilde{\sigma}^z_{#1}}

\newcommand{\Gdis}{\ensuremath{G_{\ms{dis}}}}
\newcommand{\Gshare}{\ensuremath{G_{\ms{share}}}}

\newcommand{\x}[1]{\mt{x}(#1)}

\newcommand{\bst}[1]{\mathbf{b}_{#1}}  

\newcommand{\SLE}{\mathcal{L}^{\text{ind}}}

\newcommand{\Hefftwo}{H_{\ms{eff}}^{(2\times2)}}
\newcommand{\Heffzero}{H_{\ms{eff}}^{0}}

\newcommand{\weff}{w^{\text{\tiny eff}}}

\newcommand{\HL}{\mb{H}_L^{\ms{bare}}}
\newcommand{\HR}{\mb{H}_R^{\ms{bare}}}

\newcommand{\LM}{\abbr{LM}}
\newcommand{\GM}{\abbr{GM}}
\newcommand{\MLIS}{\abbr{MLIS}}
\newcommand{\DMS}{\abbr{deg}-\abbr{MLIS}}
\newcommand{\MIC}{\abbr{dMIC}}
\newcommand{\MDC}{\abbr{dMDC}}

\newcommand{\DDD}{{\sc Dic-Dac-Doa{}}}

\newcommand{\Bc}{{\mathcal{B}}_{a}}
\newcommand{\Ba}{{\mathcal{B}}_{\ms{ang}}}

\newcommand{\PLE}{{\Pi^{\tiny{\ms{ind}}}}}

\newcommand{\Uc}{{\mathbf{U}^{\ms{clique}}}}

\def\final{1} 
\ifnum\final=1  
\newcommand{\vnote}[1]{[{\small Vicky: \bf #1}]\marginpar{*}}
\newcommand{\sidecomment}[1]{\marginpar{\tiny #1}}
\else 
\newcommand{\vnote}[1]{}
\newcommand{\sidecomment}[1]{}
\fi  

\begin{document}
\title{Limitation of Stoquastic Quantum Annealing: A Structural Perspective}

\author{
  Vicky Choi\\
  Gladiolus Veritatis Consulting Co.\footnote{\url{https://www.vc-gladius.com}}
}
\maketitle       

\begin{abstract}
We analyze the behavior of stoquastic transverse-field quantum annealing (TFQA) on a structured class of Maximum Independent Set (MIS) instances, using the same decomposition framework developed in our companion work on the \DDD{} algorithm (Beyond Stoquasticity)~\cite{Choi-Beyond}.
For these instances, we provide a structural explanation for the anti-crossing
arising from the competition between the
energies associated with a set of degenerate local minima ($\LM{}$) and the global minimum ($\GM{}$),
and analytically derive the associated exponentially small gap.
Our analysis proceeds in two steps.
First, we reduce the dynamics to an effective two-block Hamiltonian \( \Hcore \), constructed from the bare (decoupled) subsystems associated with the $\LM{}$ and $\GM{}$.
This reduction is justified analytically using the structural decomposition.
Second, we reformulate the eigenvalue problem as a generalized eigenvalue problem in a non-orthogonal basis constructed from the bare eigenstates of the subsystems.
This transformation enables a clean perturbative treatment of the anti-crossing structure, independent of the transverse field---unlike standard perturbation theory approach, which requires treating the transverse field as a small parameter.
This paper serves as a supplementary companion to our main work on the \DDD{} algorithm~\cite{Choi-Beyond}, where we demonstrate how appropriately designed non-stoquastic drivers can bypass this tunneling-induced bottleneck.
\end{abstract}

\newpage
\tableofcontents

\section{Introduction}  
This work analyzes the performance limitation of stoquastic quantum annealing (TFQA)~\cite{Farhi2000,Farhi2001,AQC-eq,AQC-Review}  on a structured class of Maximum Independent Set (MIS) instances, using the same structural framework developed in the companion work on the \DDD{} algorithm~\cite{Choi-Beyond}.
In the absence of the \( \XX \)-driver (\( \Jxx = 0 \)), 
the TFQA system evolves under a fully stoquastic Hamiltonian.
The corresponding Hamiltonian takes the form~\footnote{For completeness, we recall that the full algorithm may begin 
with an optional \emph{Stage~0}, whose sole purpose is to set the parameters 
of the problem Hamiltonian before the main evolution begins.  
The Stage~0 Hamiltonian (without the $\XX$-driver term) is
\(
\mb{H}_0(t) = \x{t} \mb{H}_{\ms{X}} + \mt{p}(t)\mb{H}_{\ms{problem}}, 
 t \in [0,1],
\)
with parameter schedules
\(
\x{t} = (1 - t)(\Gamma_0 - \Gamma_1) + \Gamma_1, 
 \mt{p}(t) = t.
\)
During this phase, the problem parameters \( w \) and \( \Jzz \)  
are ramped to their final values as \( \mt{p}(t) \) increases from 0 to 1.
Stage~0 plays no role in the present analysis and may be omitted in practice.
}
:
\[
\mb{H}(t) = \x{t} \mb{H}_{\ms{X}} + \mb{H}_{\ms{problem}},
\]
where
\(
\mb{H}_{\ms{X}} = - \sum_{i} \sigma_i^x\), and
\(
\mb{H}_{\ms{problem}} = \sum_{i \in \ver(G)} (-w_i) \shz{i} 
+ \Jzz \sum_{(i,j) \in \edge(G)} \shz{i} \shz{j}
\)
is the MIS-Ising Hamiltonian, with the shifted-\( \sigma^z \) operator \( \shz{i} := \tfrac{I + \sigma^z_i}{2} \). 
Here, we take \( w_i = w \equiv 1 \) for the unweighted MIS problem, and note that \( \Jzz \) is required to be at least \( w \). 
This is precisely the system Hamiltonian used in the \DDD{} algorithm~\cite{Choi-Beyond}, when the non-stoquastic \( \XX \)-driver is turned off, i.e., \( \Jxx = 0 \).
The annealing schedule simplifies to a single-stage evolution, i.e. there is no need to distinguish between Stage~1 and Stage~2,  with a linearly decreasing transverse field,
$\x{t} = (1 - t)\Gamma_1$.
The system Hamiltonian is stoquastic in the computational basis.

A widely observed but not fully justified limitation of stoquastic TFQA is 
the occurrence of an anti-crossing arising from the competition between the 
energies associated with a set of degenerate local minima (\LM{}) and the global minimum (\GM{}).
While such anti-crossings have long been explained using perturbative arguments 
in the small transverse-field regime (e.g.,~\cite{Amin-Choi,AKR}), we provide 
a structural approach that enables a clean perturbative treatment of the 
anti-crossing structure, independent of the transverse field, yielding an 
analytic derivation of the exponentially small gap.

In particular, we show that the relevant dynamics reduce to an effective two-block Hamiltonian \( \Hcore \), formed from bare subsystems associated with the $\LM{}$ and $\GM{}$. 
This reduction follows the same angular-momentum-based decomposition of the Hamiltonian, derived from the angular momentum structure of the cliques associated with the $\LM{}$\footnote{Although in the TFQA{} case (\( \Jxx = 0 \)) we do not explicitly identify the cliques for constructing the \(\XX\)-driver graph, the cliques associated with the critical local minima exist and are used solely for analytical purposes; they do not need to be constructively identified.} as developed in the main paper, with two additional refinements: (1) we justify that the disjoint case suffices for analysis, and (2) we apply the \( L \)-inner decomposition to derive the effective Hamiltonian $\Hcore$. For completeness, we include the necessary angular momentum decomposition in this paper.

Intuitively, \( \Hcore \) 
can be understood in terms of two embedded bare subsystems (an \( L \)-subsystem and an \( R \)-subsystem) coupled only through the empty-set state. 
Because of the additional transverse-field contribution in the \( L \)-subsystem, the system initially localizes in the \( L \)-block. 
The transition to the \( R \)-block (which supports the global minimum $\GM{}$) occurs through a tunneling-induced anti-crossing, resulting in an exponentially small gap.
While this intuition is conceptually clear, the analytic derivation is nontrivial and requires identifying the appropriate representation to capture the mechanism structurally. 
Our analysis of \( \Hcore \) is structurally guided and departs from conventional perturbation theory. 
It does not rely on a small transverse field \( \mt{x} \); instead, it reformulates the full Hamiltonian \( \Hcore \) in the basis of embedded bare eigenstates, where the overlap structure induces a generalized eigenvalue problem. 
We then express this system in a perturbative form whose perturbation term is independent of \( \mt{x} \), allowing the analysis to proceed without treating \( \mt{x} \) as a small parameter. 
Within this framework, we show that the bare energy levels provide an accurate approximation to the true spectrum of \( \Hcore \), with energy crossings between bare energy levels replaced by tunneling-induced anti-crossings.

In particular, this structure enables a perturbative treatment both away from and near the anti-crossing: we approximate the true ground energy using first-order corrections to the bare energies, and construct a reduced \( 2 \times 2 \) effective Hamiltonian on the shared two-level subspace to
derive a perturbative bound on the gap. 
The resulting bound is exponentially small in system size, confirming that the structural constraints of stoquastic TFQA lead to tunneling-induced bottlenecks that limit algorithmic performance on this family of instances. Our analysis also reveals the structural origin of tunneling-induced anti-crossings. 
In particular, once the evolution enters a regime where the system is localized around critical local minima, a tunneling-induced anti-crossing is inevitable. 

In the \( \Jxx = 0 \) case, the value of \( \Jzz \) can be made arbitrarily large without altering the qualitative dynamics, although the anti-crossing location shifts earlier as \( \Jzz \) increases. This shift occurs because larger \( \Jzz \) accelerates localization within the same-sign block, allowing the effective dynamics to be more accurately captured by \( \Hcore \) at earlier times.
However, if one sets \( \Jzz \) to be very large in the presence of an \(\XX\)-driver term with \( \Jxx > 0 \), such that \( \Jxx \ll \Jzz \),\footnote{Violating the Steering lower bound on \( \Jxx \) in \cite{Choi-Beyond}.} 
then by the cut-off Theorem~8.3 in~\cite{Choi-Beyond}, the effective Hamiltonian at the beginning of Stage~1 is already approximately \( \Hcore \), leaving no opportunity for successful structural steering.
In this case, the system localizes in the \( L \)-region and transitions to the \( R \)-region via tunneling, merely shifting the anti-crossing earlier,
as illustrated by an example in 
Figure~16 of~\cite{Choi-Beyond}. 
See also \cite{Kerman-Tunnel}, which addresses the case \( \Jzz \to \infty \).

Finally we remark that the same localization and tunneling mechanism may arise in less structured graphs as well, where degenerate local minima confine the dynamics in an analogous way.  
Through our detailed analysis, we hope to shed enough light on this anti-crossing mechanism to
provide a structural basis for a more careful assessment of claimed stoquastic speedups.

This paper is organized as follows.
In Section~\ref{sec:background}, we review the prerequisites for our analysis, including the definitions of same-sign and opposite-sign blocks, the two bipartite substructure graphs (\Gdis{} and \Gshare{}), and the definition of an anti-crossing, along with the basic matrix \( \mb{B}(w,x) \) that will be used throughout our analysis.
In Section~\ref{sec:approach}, we show that the gap analysis can be restricted to the same-sign block of the disjoint-structure case.
In Section~\ref{sec:inner}, we present the inner decompositions of the same-sign block and elaborate on the \emph{L-inner decomposition}, which reveals localization to the region associated with local minima (\emph{L-localization}), and derive the effective Hamiltonian \( \Hcore \).
In Section~\ref{sec:Hcore}, we give a detailed analysis of \( \Hcore \).
We conclude with discussion in Section~\ref{sec:end}.

\section{Background and Preliminaries}
\label{sec:background}

In this section, we review the prerequisites for our analysis:  
(1) the definitions of same-sign and opposite-sign states, sectors, and blocks (Section~\ref{sec:same-opp});  
(2) the two fundamental bipartite substructure graphs, distinguishing between the disjoint-structure case \Gdis{} and the shared-structure case \Gshare{} (Section~\ref{sec:dis-share-graphs}); and  
(3) the definition of anti-crossing together with the solution for the basic matrix \( \mb{B}(w,x) \) (Section~\ref{sec:anti-crossings}).

The material in this section is adopted from the main paper~\cite{Choi-Beyond}, and is included here
for completeness and to keep this paper self-contained.

\subsection{Same-Sign vs Opposite-Sign}
\label{sec:same-opp}
The sign structure of quantum states plays a central role in our analysis.  
Throughout this paper, all Hamiltonians are real and Hermitian, so we restrict attention to quantum states with real-valued amplitudes: each component has phase \(0\) (positive) or \(\pi\) (negative).  
\begin{mdframed}
\begin{definition}
Let \( |\psi\rangle = \sum_{x \in \mathcal{B}} \psi(x) |x\rangle \) be a quantum state with real amplitudes in a basis \( \mathcal{B} \).

\begin{itemize}
    \item \( |\psi\rangle \) is called a \textbf{same-sign state} if \( \psi(x) \geq 0 \) for all \( x \in \mathcal{B} \).  
    That is, all components are in phase (with relative phase \( 0 \)).

    \item \( |\psi\rangle \) is called an \textbf{opposite-sign state} if there exist \( x, x' \in \mathcal{B} \) such that \( \psi(x) > 0 \) and \( \psi(x') < 0 \).  
    In this case, some components are out of phase, differing by a relative phase of \( \pi \).
\end{itemize}
\end{definition}
\end{mdframed}
Unless stated otherwise, we take \( \mathcal{B} \) to be the computational basis when referring to same-sign or opposite-sign states.
More generally, the computational basis is assumed whenever the basis is unspecified.

Accordingly, we define the notions of same-sign and opposite-sign bases, sectors, and blocks:
\begin{definition}
An orthonormal basis consisting entirely of same-sign states is called a same-sign basis.  
A basis that includes at least one opposite-sign state is called an opposite-sign basis.
A subspace spanned by a same-sign basis is called a same-sign sector; otherwise, it is called an opposite-sign sector.
A submatrix (or block) of a Hamiltonian is called a same-sign block if it is expressed in a same-sign basis.  
Otherwise, it is referred to as an opposite-sign block.
\end{definition}

\paragraph{Example.}
The state \( \ket{+} = \tfrac{1}{\sqrt{2}}(\ket{0} + \ket{1}) \) is a same-sign state,  
while \( \ket{-} = \tfrac{1}{\sqrt{2}}(\ket{0} - \ket{1}) \) is an opposite-sign state.  
The computational basis \( \{ \ket{0}, \ket{1} \} \) is a same-sign basis of \( \mb{C}^2 \),  
and the Hadamard basis \( \{ \ket{+}, \ket{-} \} \) is an opposite-sign basis.

These definitions connect directly to stoquasticity.  
By the Perron--Frobenius theorem, the ground state of a stoquastic Hamiltonian (i.e., one with non-positive off-diagonal entries in a given basis) is a same-sign state in that basis~\cite{AQC-eq,general-PF2}.  
In particular, when expressed in the computational basis, the ground state of a stoquastic Hamiltonian is necessarily same-sign.  
By contrast, a non-stoquastic Hamiltonian may have a ground state that is either same-sign (\emph{Eventually Stoquastic}) or opposite-sign (\emph{Proper Non-stoquastic})~\cite{Choi2021}.  
In this paper we focus on the stoquastic case, corresponding to \( \Jxx = 0 \).

\subsection{Two Fundamental Bipartite Substructure Graphs}
\label{sec:dis-share-graphs}
Recall that an independent set is \emph{maximal} if no larger independent set contains it.  
Each maximal independent set corresponds to a local minimum of the energy function.  
A collection of maximal independent sets (\MLIS{}) all having the same size \( m \)  
corresponds to a set of \emph{degenerate} local minima with equal energy \( -m \).  
When the meaning is clear (i.e., they all have the same size), we refer to such a collection simply as a \DMS{},  
and its cardinality as the \emph{degeneracy}.  
In this work, we use the terms \emph{degenerate local minima} and \DMS{} interchangeably.

\subsection*{\MIC{}: Structured Form of a \DMS{}}
In its simplest structured form, a \DMS{} can be represented as a \MIC{}, 
namely a collection of mutually independent cliques whose vertices together generate exactly all the \MLIS{} in that \DMS{}.

\begin{definition}
A \MIC{} of size \( k \) consists of \( k \) \emph{independent cliques} (i.e., no edges exist between them),  
denoted as \( \text{Clique}(w_i, n_i) \),  
where \( w_i \) is the vertex weight and \( n_i \) is the clique size, for \( 1 \leq i \leq k \).  
Each maximal independent set in the corresponding \DMS{} is formed by selecting exactly one vertex from each clique.  
In this case, the degeneracy of the \MIC{} is given by
\(
\prod_{i=1}^{k} n_i.
\)
\end{definition}

When the cliques are \emph{dependent}---that is, edges may exist between them---we refer to the structure as an \MDC{},  
a \DMS{} formed by a set of dependent cliques.  
This represents a relaxed structure to which our analysis may be generalized.  

As explained in the main paper, we reduce the analysis to
a sequence of simpler bipartite substructures,  
each formed by a critical \MIC{} and the global maximum (\GM{}). 
 
\begin{mdframed}
Recall that each \MIC{} corresponds to a set of degenerate local minima (\LM{}) in the MIS--Ising energy landscape. 
The terms \LM{} and \MIC{} are thus used interchangeably. 
We use \GM{} to refer both to the (global) maximum independent set and to its corresponding global minimum in the energy landscape. 
In what follows, we use \LM{} and \GM{} to denote the \MIC{} and \GM{} in each such bipartite substructure, respectively.
\end{mdframed}


We consider two fundamental bipartite substructures:
\begin{itemize}
    \item \textbf{\Gdis}: The local minima (\LM{}) and the global minimum (\GM{}) are vertex-disjoint.
    \item \textbf{\Gshare}: The \GM{} shares exactly one vertex with each clique in the \LM{}.
\end{itemize}

We begin with the \emph{disjoint-structure graph} \( \Gdis = (V, E) \),  
in which the vertex set \( V \) is partitioned into left and right (disjoint) vertex sets,  
with the following structural properties:
\begin{itemize}
  \item The left component is defined by a set \( L = \{ C_1, \dots, C_{m_l} \} \) of \( m_l \) disjoint cliques,  
        each denoted \( C_i = \cl(w_i, n_i) \).  
        We let \( V_L = \bigcup_{i=1}^{m_l} C_i \) denote the full vertex set.
  \item The right component \( R \) consists of \( m_r \) independent vertices, each with weight \( w_r \). 
  \item Every vertex in \( V_L \) is adjacent to every vertex in \( R \).
\end{itemize}

In this paper we mainly focus on the unweighted MIS case,  
and assume uniform weights \( w_i = w_r = w \).  
Under the MIS--Ising mapping with these uniform weights,  
\( V_L \) corresponds to the degenerate local minima (\LM{}) with degeneracy \( \prod_{i=1}^{m_l} n_i \),  
while \( R \) defines the global minimum (\GM{}) with \( m_g = m_r \).
Each local minimum (in \LM{}) corresponds to a maximal independent set of size \( m_l \), and thus has energy \( -m_l \).  
The global minimum (\GM{}) has energy \( -m_g \).

We now define the \emph{shared-structure graph} \( \Gshare \),  
which differs from \( \Gdis \) in that each vertex in \( R \) is adjacent to all but one vertex in each clique of \( L \).  
This modification allows the \GM{} to include shared vertices from the cliques in \( L \),  
thereby introducing overlap with the \LM{}.
Structurally, \( L \) and \( R \) are defined exactly as in \( \Gdis \),  
but with the adjacency rule modified as above.  
Specifically, the global maximum \GM{} consists of one shared vertex from each clique \(C_i \in L\) together with all \(m_r\) independent vertices in \(R\), yielding a total size \(m_g = m_l + m_r\).
Figure~\ref{fig:dis-share-graph} illustrates both cases. For convenience, we write \( m := m_l \), dropping the subscript $l$ when no confusion arises.  
We assume \( \sum_{i=1}^{m} \sqrt{n_i} > m_g \),
so that an anti-crossing is induced by the competition between \( \LM \) and \( \GM \).

\begin{figure}[h]
  \centering
  \begin{subfigure}[b]{0.36\textwidth}
    \includegraphics[width=\textwidth]{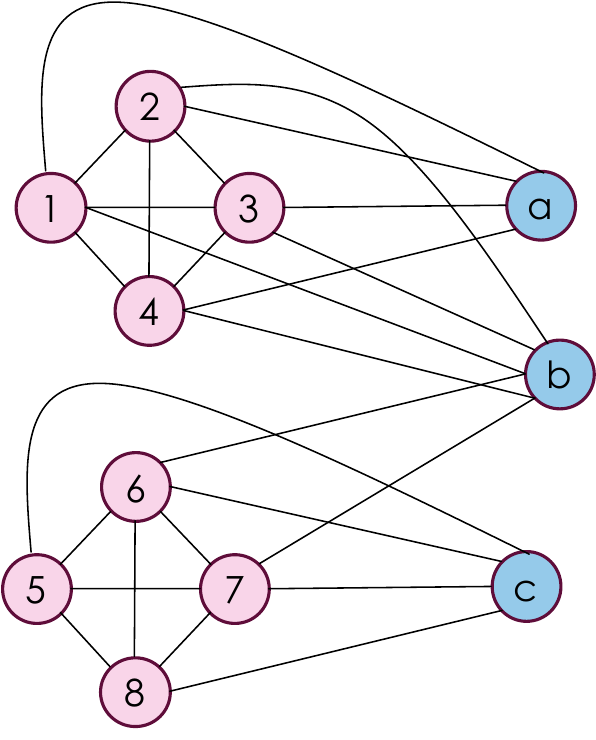}
    \caption{The \LM{} and \GM{} are vertex-disjoint.}
  \end{subfigure}
  \hfill
  \begin{subfigure}[b]{0.36\textwidth}
    \includegraphics[width=\textwidth]{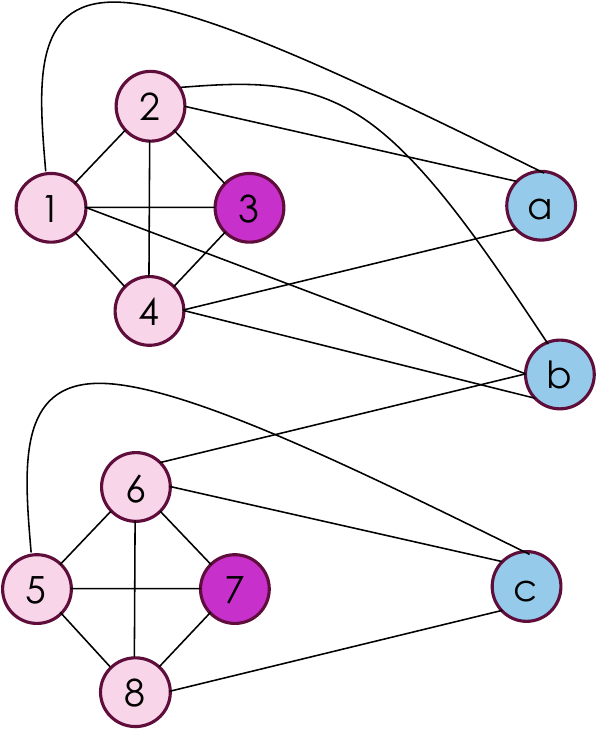}
    \caption{The \GM{} shares exactly one vertex with each clique in the \LM{}.}
  \end{subfigure}
\caption{
Example graphs illustrating the \Gdis{} and \Gshare{} structures.  
Recall that each \LM{} here is structurally a \MIC{}.  
(a) Disjoint-structure graph \Gdis: The set \( L \) consists of \( m_l = 2 \) disjoint cliques,  
each of size \( n_1 = n_2 = 4 \), with their vertices (pink) forming the local minima (\LM{}).  
The set \( R \) (blue) consists of \( m_r = 3 \) independent vertices,  
forming the global minimum (\GM{}).\\  
(b) Shared-structure graph \Gshare: The set \( L \) again consists of two cliques of size \( n_1 = n_2 = 4 \),  
with pink and purple vertices. The purple vertices (one per clique) are shared between the \LM{} and the \GM{}.  
The set \( R \) (blue) contains \( m_r = 3 \) independent vertices.  
The global minimum consists of all vertices in \( R \), together with the shared purple vertices in \( L \), giving \( m_g = 5 \).  
In both cases, edges between the pink vertices in \( L \) and all vertices in \( R \) are complete,  
though not all are shown for visual clarity.
}
\label{fig:dis-share-graph}
\end{figure}

\subsection{Anti-Crossing and the Two-Level Hamiltonian \( \mb{B}(w,x) \)}
\label{sec:anti-crossings}

The term \emph{anti-crossing} is sometimes used loosely, so we begin with a precise notion in the two-level case, then extend it to multi-level systems.
We also introduce a canonical two-level Hamiltonian whose eigensystem will be used throughout our analysis. 

\subsubsection{Level Repulsion vs.~Anti-Crossing}
\label{subsec:level-repulsion}
We begin by distinguishing the concept of an \emph{anti-crossing} (also called \emph{avoided-crossing}) from level repulsion.

Consider a generic two-level Hamiltonian of the form
\[
\mb{H}(x) :=
\begin{bmatrix}
e_1(x) & v(x) \\
v(x) & e_2(x)
\end{bmatrix},
\]
where \( e_1(x) \), \( e_2(x) \), and \( v(x) \) are real-valued functions of a parameter \( x \). 

The eigenvalues of this Hamiltonian are
\(
\lambda_{\pm}(x) = \tfrac{e_1(x) + e_2(x)}{2} \pm \tfrac{1}{2} \sqrt{(e_1(x) - e_2(x))^2 + 4v(x)^2},
\)
and the energy gap between them is
\(
\Delta(x) := \lambda_+(x) - \lambda_-(x) = \sqrt{(e_1(x) - e_2(x))^2 + 4v(x)^2}.
\)
The off-diagonal term \( v(x) \) induces \emph{level repulsion}: if \( v(x) \neq 0 \), then the eigenvalues never cross, and the gap \( \Delta(x) \) remains strictly positive.
Thus, assuming the off-diagonal coupling \( v(x) \) is nonzero, level repulsion is always present.

\begin{definition}
We say that an \emph{anti-crossing} occurs when the two
unperturbed energy levels $e_1(x)$ and $e_2(x)$ cross, i.e.,
$e_1(x_*) = e_2(x_*)$ for some $x_*$, and the off-diagonal coupling
$v(x_*) \neq 0$. 
In this case the eigenvalue curves form an anti-crossing with gap
\(
\Delta_{\min} = 2|v(x_*)|.
\)
\end{definition}
The size of the anti-crossing gap depends on $|v(x_*)|$: stronger coupling leads to a larger gap, while weaker coupling results in a narrower one.

By contrast, if the two diagonal entries \( e_1(x) \) and \( e_2(x) \) remain well separated for all \( x \), then the system exhibits level repulsion but not an anti-crossing. Figure~\ref{fig:B2wx} illustrates an example of level repulsion without an anti-crossing. 

The eigenvectors of the two-level Hamiltonian are given by
\[
\ket{\lambda_{-}(x)} = \cos\theta(x) \ket{0} + \sin\theta(x) \ket{1}, \quad
\ket{\lambda_{+}(x)} = -\sin\theta(x) \ket{0} + \cos\theta(x) \ket{1},
\]
where the mixing angle \( \theta(x) \) satisfies
\(
\tan(2\theta(x)) = \tfrac{2v(x)}{e_1(x) - e_2(x)}.
\)
Thus, near the anti-crossing point \( x = x_* \), the eigenstates interpolate between the unperturbed basis states.

\begin{remark}
The trigonometric expression for eigenvectors in terms of the mixing angle \( \theta(x) \),
is equivalent to the rational-form representation
\[
\ket{\lambda_{-}(x)} = \tfrac{1}{\sqrt{1 + \gamma(x)^2}} \left(\gamma(x) \ket{0} +  \ket{1} \right),\quad
\ket{\lambda_{+}(x)} = \tfrac{1}{\sqrt{1 + \gamma(x)^2}} \left( \ket{0} -\gamma(x) \ket{1} \right),
\]
where the two parametrizations are related by
\(
\gamma(x) = \tfrac{1}{\tan\theta(x)}, \text{and} \quad \tan(2\theta(x)) = \tfrac{2v(x)}{e_1(x) - e_2(x)}.
\)
This rational-form expression is particularly useful for our analysis, as it aligns directly with the basic matrix form introduced below.
\end{remark}

Our earlier work~\cite{Choi2021,Choi2020}, including the development of the original \DDD{} algorithm, was motivated by investigating the structural characteristics of eigenstates around the anti-crossing.

\subsubsection{Anti-crossing in a Multi-level System}
In a multi-level system, the notion of an anti-crossing extends naturally
by restricting the Hamiltonian to the two-dimensional subspace spanned by
the pair of eigenstates whose unperturbed energies intersect.
This reduction yields a $2 \times 2$ effective Hamiltonian that captures the
essential structure of the anti-crossing, including both the energy gap and
the interpolating behavior of the eigenstates.
Thus, the same framework as in the two-level case applies.

With this perspective, we refine the definition of an $(L,R)$-anti-crossing
given in \cite{Choi2021}. Recall that $E_0^A(t)$ denotes the ground state
energy of the Hamiltonian $\mb{H}_A(t)$ projected to the subspace spanned by
the subsets of $A$.
In particular, $E_0^{\LM}(t)$ and $E_0^{\GM}(t)$ denote the bare energies 
associated with $\LM$ (where $A = \bigcup_{M \in \LM} M$) and $\GM$, respectively.

\begin{definition}
  We say an anti-crossing is an $(L,R)$-anti-crossing at $t_*$ if there exist
  bare energies $E_0^L(t)$ and $E_0^R(t)$ such that:
  \begin{enumerate}
    \item $E_0^L(t)$ and $E_0^R(t)$ approximate the unperturbed energy levels
    of the effective $2 \times 2$ Hamiltonian describing the anti-crossing
    for $t \in [t_* - \delta, t_* + \delta]$ for some small $\delta>0$; and
    \item $E_0^L(t)$ and $E_0^R(t)$ cross at $t_*$, i.e.\ $E_0^L(t_*) = E_0^R(t_*)$.
  \end{enumerate}
\end{definition}
See Figure~\ref{fig:AC} for an illustration.

\begin{figure}[!htbp]
  \centering
  \includegraphics[width=0.5\textwidth]{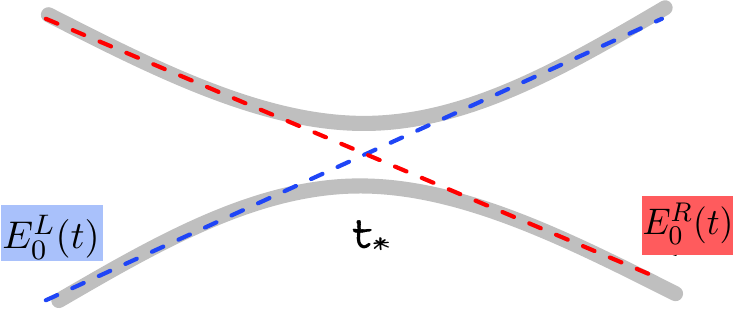}
  \caption{
Schematic of an $(L,R)$-anti-crossing.
Dashed lines: bare energies $E_0^L(t)$ and $E_0^R(t)$ crossing at $t_*$.
Solid gray curves: the two lowest eigenvalues of the Hamiltonian, showing
the avoided crossing that originates from this bare crossing.
  }
  \label{fig:AC}  
\end{figure}


\begin{remark}
In the companion work~\cite{Choi-Beyond}, the absence of an anti-crossing can 
be inferred directly from the non-crossing of the corresponding bare energy 
levels, without explicitly constructing the effective \( 2 \times 2 \) 
Hamiltonian. By contrast, in the present work we evaluate the anti-crossing gap 
by constructing the effective Hamiltonian explicitly.
\end{remark}

\begin{remark}
In~\cite{JarretJordan2014}, it was observed that an exponentially small gap occurs ``if and only if a ground state consists of two `lobes' with exponentially small amplitude in the region between them.'' 
This two-lobe structure corresponds naturally to the two arms of an $(L,R$)-anti-crossing. 
[As shown in the right panel of Figure~\ref{fig:anticrossing-combined}, the ground state undergoes a sharp exchange at the anti-crossing point. 
If one plots the ground state wavefunction at the anti-crossing, it exhibits two lobes: one localized on the $L$ region and the other on $R$ region (with the suppressed region between them quantified by the overlap \( g_0 = \braket{\bar{L}_0 | \bar{R}_0} \)).]
Finally, we remark that while the Cheeger inequalities (used, e.g., in~\cite{AQC-eq} to give a large lower bound on the spectral gap) appear to depend only on the ground state wavefunction, a closer examination reveals that it actually bound the first excited state implicitly.
In this sense, the anti-crossing gap size derived from a \( 2 \times 2 \) effective Hamiltonian is more descriptive and also gives a tighter bound.
\end{remark}

\subsubsection{The Basic Matrix \( \mb{B}(w,x) \): Eigenvalues and Eigenstates}
\label{subsec:basic-matrix}

We define the following effective two-level Hamiltonian, which will serve as a basic building block for our analysis throughout the paper:
\begin{equation}
\mb{B}(w, x) :=
\begin{bmatrix}
-w & -\tfrac{1}{2}x \\
-\tfrac{1}{2}x & 0
\end{bmatrix},
\label{eq:B-matrix}
\end{equation}
where \( w = w(t) \) and \( x = x(t) \) are real-valued parameters, typically derived from problem Hamiltonians and driver strengths. This is a special case of a spin-\( \tfrac{1}{2} \) system, with analytic eigenstructure.
The eigenvalues of $\mb{B}(w,x)$ are
\begin{equation}
\beta_k = -\tfrac{1}{2}\!\left(w + (-1)^k \sqrt{w^2 + x^2}\right), \quad k=0,1,
\label{eq:B-evals}
\end{equation}
with normalized eigenvectors
\begin{equation}
\ket{\beta_0} = \tfrac{1}{\sqrt{1+\gamma^2}} \bigl(\gamma\ket{0} + \ket{1}\bigr), 
\quad
\ket{\beta_1} = \tfrac{1}{\sqrt{1+\gamma^2}} \bigl(\ket{0} - \gamma\ket{1}\bigr),
\label{eq:B-evecs}
\end{equation}
where the mixing coefficient is
\(
\gamma = \tfrac{x}{w+\sqrt{w^2+x^2}}.
\)

Figure~\ref{fig:B2wx} visualizes the eigenspectrum and ground state behavior under variation of \( x \) and \( w \).


\begin{figure}[h!]
  \centering

    \includegraphics[width=0.58\textwidth]{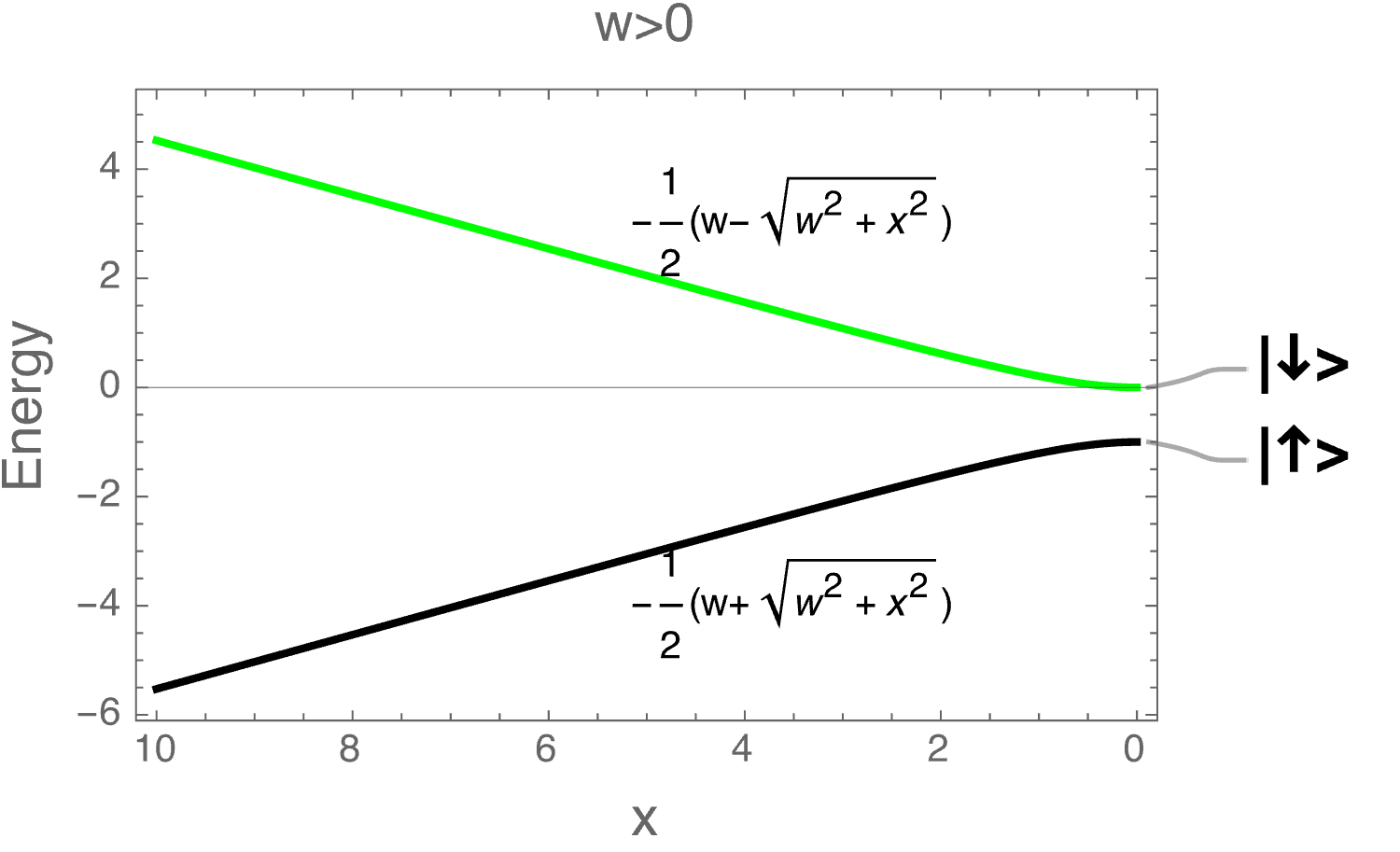}

  \caption{Eigenvalues of the two-level Hamiltonian \( \mb{B}(w, x) \), where \( \beta_k = -\tfrac{1}{2} \left( w + (-1)^k \sqrt{w^2 + x^2} \right), \; k = 0, 1 \). $w=1$.
    The ground state energy (\( \beta_0 \)) is shown in black, and the first excited state energy (\( \beta_1 \)) is shown in blue.
    The energy gap widens as \( x \) increases---there is no anti-crossing in this case.
    The ground state is 
    \( \ket{\beta_0} = \tfrac{1}{\sqrt{1 + \gamma^2}} \left(\gamma  \ket{0} + \ket{1} \right) \), with \( \gamma = \tfrac{x}{w + \sqrt{w^2 + x^2}} \). \( \ket{0} = \ket{\uparrow} \) and \( \ket{1} = \ket{\downarrow} \).}
  \label{fig:B2wx}
\end{figure}

\section{Reduction to Same-Sign Block in the Disjoint-Structure Case}
\label{sec:approach}
In this section, we establish that it is sufficient to restrict our analysis to 
the same-sign block and to focus on the disjoint case.
We first express the full Hamiltonian in block form with respect to the
    low- and high-energy subspaces of $\mc{V}_L$, and justify that for our
    purposes it suffices to analyze the projected Hamiltonian
    $\mb{H}^{\ms{low}}$ acting on $\mc{L}_{-}$ in Section~\ref{sec:reduct-low}.
    Then we describe the angular-momentum based decomposition of $\mb{H}^{\ms{low}}$ in
    in Section~\ref{sec:am-decomp}.

\subsection{Reduction to the Analysis of the Low-Energy Hamiltonian}
\label{sec:reduct-low}
Let $\mc{V}_L$ be the Hilbert space of all vertices in the union of cliques in $L$.
This space decomposes into:
\begin{itemize}
    \item a \emph{low-energy subspace}, spanned by all independent-set states within $L$;
    \item a \emph{high-energy subspace}, spanned by all dependent-set states within $L$.
\end{itemize}
We define
\[
\mc{L}_{-} := (\text{low-energy subspace of }\mc{V}_L) \otimes \mc{V}_R,
\quad
\mc{L}_{+} := (\text{high-energy subspace of }\mc{V}_L) \otimes \mc{V}_R.
\]
Here, $\mc{L}_{+}$ consists of states containing at least one intra-clique edge
and thus incurring the $\Jzz^{\ms{clique}}$ penalty.

Let $\Pi_{-}$ and $\Pi_{+}$ be the orthogonal projectors onto
$\mc{L}_{-}$ and $\mc{L}_{+}$, respectively.
With respect to this decomposition, the system Hamiltonian can be written in block form:
\[
  \mb{H} \;=\;
  \begin{pmatrix}
    \mb{H}^{\ms{low}} & \mb{V} \\
    \mb{V}^\dagger & \mb{H}^{\ms{high}}
  \end{pmatrix},
\]
where $\mb{H}^{\ms{low}}$ and $\mb{H}^{\ms{high}}$ are the projections into
$\mc{L}_{-}$ and $\mc{L}_{+}$, respectively.

In the main paper, we showed that by the end of Stage~0, if
$\Jzz^{\ms{clique}}$ is chosen sufficiently large so that all states in
$\mc{L}_{+}$ lie well above those in $\mc{L}_{-}$ in energy, the Stage~1
evolution is governed exactly by the restricted Hamiltonian
$\Heff = \mb{H}^{\ms{low}}$.
The subsequent analysis can therefore focus entirely on $\mb{H}^{\ms{low}}$.

In the present $\Jxx = 0$ setting, we have not identified the $\XX$-driver
graph explicitly, and we set $\Jzz^{\ms{clique}} = \Jzz$, so we cannot
assume an energetic cut-off between $\mc{L}_{+}$ and $\mc{L}_{-}$.
Nevertheless, removing higher-energy subspaces from a Hamiltonian cannot
\emph{decrease} the spectral gap, since doing so can only eliminate potential
low-lying excitations.
Consequently, if $\mb{H}^{\ms{low}}$ already has a small gap, the full
Hamiltonian $\mb{H}$ can have a gap no larger.
Because the ground state lies in $\mc{L}_{-}$ prior to the anti-crossing,
we have
\(
  \Delta(\mb{H}) \ \le\ \Delta(\mb{H}^{\ms{low}}).
\)
This allows us to use $\mb{H}^{\ms{low}}$ to obtain an \emph{upper bound}
on the spectral gap of the full Hamiltonian.
Thus, in both cases---whether with or without the $\XX$-driver---the
subsequent analysis focuses on $\Heff = \mb{H}^{\ms{low}}$.

Having reduced our focus to the low-energy subspace $\mc{L}_{-}$, 
we next construct an angular momentum basis that reveals the block 
structure of $\mb{H}^{\ms{low}}$ and sets up the tunneling analysis.

\subsection{Angular-Momentum Based Decomposition of the $\mb{H}^{\ms{low}}$}
\label{sec:am-decomp}
Our construction starts from the local angular momentum structure of a single 
clique, restricted to the low-energy subspace (see also ~\cite{Kerman-Tunnel}). 
We then extend this to a global angular-momentum basis for $\mc{L}_{-}$, 
in which $\mb{H}^{\ms{low}}$ becomes block-diagonal, separating 
\emph{same-sign} and \emph{opposite-sign} blocks.

The construction proceeds hierarchically:
\begin{itemize}
    \item \textbf{Single-clique decomposition.}  
    We analyze the low-energy subspace of a single clique in detail, 
    introducing its total angular momentum basis, identifying the 
    same-sign and opposite-sign states, and deriving the block structure 
    of the restricted Hamiltonian (Section~\ref{sec:single-clique}).
    
    \item \textbf{Global block decomposition.}
    By tensoring the single-clique bases over all cliques in $L$ and
    combining with $\mc{V}_R$, we obtain a global basis in which
    $\mb{H}^{\ms{low}}$ takes a block-diagonal form (Section~\ref{sec:block-low}).
    \end{itemize}

\subsubsection{Single-clique Decomposition}
\label{sec:single-clique}

In this section, we focus on analyzing the Hilbert space associated with a single clique.
We begin with a single clique, since each clique in $L$ exhibits the 
same local angular momentum structure; this will serve as the building 
block for the global block decomposition in Section~\ref{sec:block-low}.

To fix notation, let \( G_c = (V, E) \) be a clique,  
where \( V = \{1, \dots, n_c\} \) is the set of vertices, and  
\( E = \{(i,j) : i < j, \; i, j \in V\} \) is the set of edges.  
We assume all vertices in the clique have the same weight \( w_c \),  
i.e., \( w_i = w_c \) for all \( i \in V \).

The Hilbert space \( \mathcal{V}_c \) for a clique of size \( n_c \) consists of \( 2^{n_c} \) computational basis states, each corresponding to a binary string of length \( n_c \). Among these, only \( n_c + 1 \) basis states correspond to independent sets:
\[
\begin{array}{l}
\bst{0} = \texttt{00\ldots0} \\
\bst{1} = \texttt{10\ldots0} \\
\bst{2} = \texttt{01\ldots0} \\
\quad\vdots \\
\bst{n_c} = \texttt{00\ldots1}
\end{array}
\]
with 
\(
\mathbb{N}_{\text{ind}} = \left\{
\bst{0},
\bst{1},
\ldots,
\bst{n_c}
\right\},
\)
where each \( \bst{i} \) is a binary string of length \( n_c \) with a single 1 in position \( i \) (and 0s elsewhere), and \( \bst{0} \) is the all-zero string.

The energy associated with each singleton state is \( -w_c \), while the empty set has energy 0.  
In contrast, the energy of all other bit strings---which correspond to dependent sets---is at least \( (\Jzz^{\text{clique}} - 2w_c) \).  
Hence, the Hilbert space admits a natural decomposition:
\begin{equation}
  \label{eq:clique-decomposition}
\mathcal{V}_c = \SLE \oplus \mathcal{H}^{\text{dep}},
\end{equation}
where \( \SLE \) is the low-energy subspace spanned by \( \mathbb{N}_{\text{ind}} \),  
and \( \mathcal{H}^{\text{dep}} \) is the high-energy subspace spanned by the dependent-set states.

In the following, we identify how \( \SLE \) decomposes into distinct angular momentum sectors, describe the corresponding block structure of the restricted Hamiltonian,  
and provide exact expressions for its eigenvalues and eigenvectors.  

\subsubsection*{Low-Energy Subspace in Total Angular Momentum Basis $\Bc$}

\begin{mdframed}
\begin{lemma}[Lemma 6.1 in \cite{Choi-Beyond}]
  \label{single-clique-basis}
  The total angular momentum basis for \( \SLE \) consists of the states:
  \begin{align}
  \label{eq:Bc}
  (\Bc) \left\{
  \begin{array}{llll}
    \ket{s,-(s-1)}, &\ket{1,(s-1),-(s-1)},&\ldots&\ket{n_c-1,(s-1),-(s-1)}\\
    \ket{s, -s} & & &
  \end{array}
  \right.
  \end{align}
  where \( s = \tfrac{n_c}{2} \) is the total spin.

  Explicitly:
  \begin{itemize}
    \item \( \ket{s, -s} = \ket{\bst{0}} \), representing the empty set.
    \item \( \ket{s,-(s-1)} = \tfrac{1}{\sqrt{n_c}} \sum_{i=1}^{n_c} \ket{\bst{i}} \), a uniform superposition of all singletons with positive amplitudes.
    \item \( \ket{k, s-1, -(s-1)} \), for \( k = 1, \ldots, n_c - 1 \), consists of a superposition of singleton states with both positive and negative amplitudes.
  \end{itemize}
  Thus, \( \ket{s, -s} \) and \( \ket{s,-(s-1)} \) are same-sign states, while \( \ket{k, s-1, -(s-1)} \) are opposite-sign states.
\end{lemma}
\end{mdframed}

\paragraph{Remark (Basis Reordering).}  
For convenience, we reorder the basis states in \( \Bc \) as follows:
\[
(\Bc')\quad \ket{s, -(s - 1)},\, \ket{s, -s},\, \ket{1, s - 1, -(s - 1)},\, \dots,\, \ket{n_c - 1, s - 1, -(s - 1)}.
\]
That is, we swap the order of the two same-sign basis states.
This ordering simplifies the representation of operators in the next steps.

\paragraph{Basis Transformation.}  
The transformation between the computational basis \( \{ \ket{\bst{i}} \} \) and the angular momentum basis $(\Bc')$ can be derived either from the Clebsch--Gordan coefficients or directly from the relationships established in the proof. Specifically:
\begin{itemize}
  \item The state \( \ket{s, -(s - 1)} \) is a uniform superposition over all singleton states \( \{ \ket{\bst{i}} \}_{i=1}^{n_c} \).
  \item The remaining states \( \ket{k, s - 1, -(s - 1)} \), for \( k = 1, \dots, n_c - 1 \), form an orthogonal complement to \( \ket{s, -(s - 1)} \) within the subspace spanned by \( \{ \ket{\bst{1}}, \dots, \ket{\bst{n_c}} \} \).
\end{itemize}

We denote the basis transformation matrix from the computational basis to the angular momentum basis by \( \Uc \).

Although the present analysis will later specialize to the 
$\Jxx = 0$ case, we retain the $\Jxx$ term throughout the 
single-clique derivation. This ensures that the resulting 
expressions apply to the general setting and remain consistent 
with the corresponding formulas in the companion paper.

\subsubsection*{Spin Operators in the \(\Bc'\) Basis}
Consider the spin operators on \( \mathcal{V}_c \):
\[
\Sop{\Z} = \tfrac{1}{2}\sum_{i=1}^{n_c} \sigma_i^z, \quad
\Sop{\sZ} = \sum_{i=1}^{n_c} \shz{i}, \quad
\Sop{\X} = \tfrac{1}{2}\sum_{i=1}^{n_c} \sigma_i^x, \quad
\Sop{\XX} = \tfrac{1}{4} \sum_{ij \in \edge(G_{\text{driver}})} \sigma_i^x \sigma_j^x,
\]
where \( G_{\text{driver}} = G_c \).

We project these operators onto the low-energy subspace \( \SLE \) using the projection operator \( \PLE \), and then transform them into the \( \Bc' \) basis via the basis transformation \( \Uc \). For any operator \( \mathrm{X} \), we use a bar to denote the operator:
\[
\bar{\mathrm{X}} = \Uc^{\dagger} \PLE \mathrm{X} \PLE \Uc.
\]

\begin{mdframed}
  \begin{theorem}[Theorem 6.2 in \cite{Choi-Beyond}]
  \label{thm:Sop}
  The restricted operators in the \(\Bc'\) basis are block-diagonal and given explicitly by:
  \begin{align}
    \left\{
    \begin{array}{ll}
    & \bar{\Sop{\sZ}} = \shz{} \oplus 1 \oplus \cdots \oplus 1, \\[5pt]  
    & \bar{\Sop{\X}} =  \tfrac{\sqrt{n_c}}{2} \sigma^x \oplus 0 \oplus \cdots \oplus 0, \\[5pt]
      & \bar{\Sop{\XX}} = \left(\tfrac{n_c - 1}{4}\right)\shz{} \oplus \left(-\tfrac{1}{4}\right) \oplus \cdots \oplus \left(-\tfrac{1}{4}\right).
    \end{array}
    \right.
    \label{eq:transformed-operators}
  \end{align}
  where \( \shz{} \) and \( \sigma^x \) act on the effective spin-\( \tfrac{1}{2} \) (two-dimensional same-sign) subspace, while the scalars act on spin-0 (one-dimensional opposite-sign) subspaces.
\end{theorem}
\end{mdframed}

\subsubsection*{Decomposition into Spin-\(\tfrac{1}{2}\) and Spin-$0$ Components}
\label{sec:spin-structure}
The transformed operators given by Theorem~\ref{thm:Sop} offer a clear physical interpretation:
the low-energy subspace \( \SLE \) decomposes into a direct sum consisting of a single effective spin-\(\tfrac{1}{2}\) subsystem,
together with \( n_c - 1 \) (opposite-sign) spin-$0$ subsystems:
\begin{align}
  \SLE = \left[\tfrac{1}{2}\right]_{n_c} \oplus \underbrace{0 \oplus \cdots \oplus 0}_{n_c - 1}.
  \label{eq:single-clique}
\end{align}

The effective spin-\(\tfrac{1}{2}\) subsystem \( \left[\tfrac{1}{2}\right]_{n_c} \) is spanned by the two same-sign basis states:
\(
\left|\tfrac{1}{2}, -\tfrac{1}{2}\right\rangle = |s, -(s-1)\rangle, \quad 
\left|\tfrac{1}{2}, \tfrac{1}{2}\right\rangle = |s, -s\rangle.
\)
The spin-$0$ components correspond to the opposite-sign basis vectors:
\(
|k, s - 1, -(s - 1)\rangle,\) for \(k = 1, \dots, n_c - 1.
\)

Correspondingly, the Hamiltonian decomposes into a same-sign two-dimensional effective spin-\(\tfrac{1}{2}\) block and \( n_c - 1 \) opposite-sign spin-0 blocks.  
The system Hamiltonian \( \mb{H}_1 \), which governs the evolution during Stages~1 and~2, when restricted to the low-energy subspace \( \SLE \) of the clique \( G_c \), takes the form
\[
\bar{\mb{H}_1} 
= -\mt{x}\, \bar{\Sop{\X}}
+ \mt{jxx}\, \bar{\Sop{\XX}} 
- w_c\, \bar{\Sop{\sZ}}.
\]

Substituting the operator expressions from Eq.~\ref{eq:transformed-operators} in Theorem~\ref{thm:Sop}, we obtain
\begin{align}
\bar{\mb{H}_1} &= 
\left(
- \tfrac{\sqrt{n_c}}{2}\, \mt{x}\, \sigma^x
+ \left( - w_c + \tfrac{n_c - 1}{4} \mt{jxx} \right)\shz{}
\right)
\oplus
\underbrace{
\left[
- \left( w_c + \tfrac{1}{4} \mt{jxx} \right)
\right]
\oplus \cdots \oplus
\left[
- \left( w_c + \tfrac{1}{4} \mt{jxx} \right)
\right]
}_{n_c - 1}\\
&= \mb{B}(\mt{\weff_c}, \sqrt{n_c}\mt{x})
\oplus 
\underbrace{[-(w_c + \tfrac{1}{4} \mt{jxx})] 
\oplus \cdots 
\oplus [- (w_c + \tfrac{1}{4} \mt{jxx})]}_{n_c - 1}
\label{eq:Hsc}
\end{align}
where the effective weight is defined as
\(
\mt{\weff_c} = w_c - \tfrac{n_c - 1}{4} \mt{jxx}.
\)

An illustration of this basis transformation is provided in Figure~\ref{fig:M4},
which shows how the original product basis
is transformed into a direct-sum decomposition via total angular momentum.

\begin{figure}[h!]
  \centering
  \includegraphics[width=0.8\textwidth]{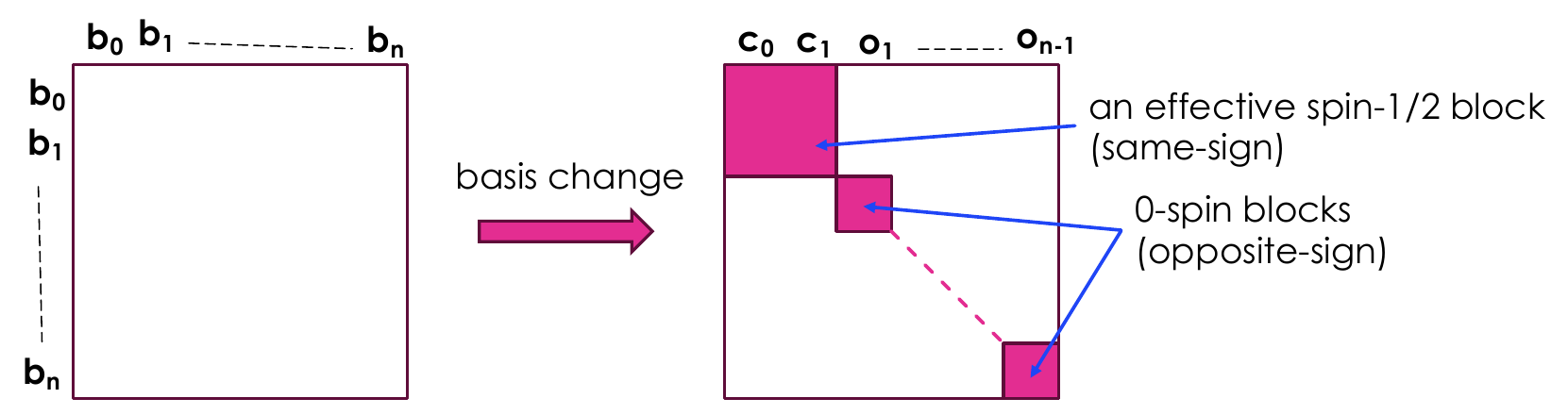}
  \caption{
    Basis transformation of \( \SLE \) from the product space  \( \{ \bst{0}, \bst{1}, \dots, \bst{n} \} \) to the direct-sum space (angular momentum basis).
    The same-sign states \( \{{\bf c_0, c_1} \} \) are the basis of the effective spin-\(\tfrac{1}{2}\) subspace, while the opposite-sign states \( \{ {\bf o_1, \dots, o_{n - 1} }\} \)
    are the bases of the spin-$0$ subspaces. The Hamiltonian matrix decomposes accordingly into a same-sign block and \( n - 1 \) opposite-sign blocks.
  }
  \label{fig:M4}
\end{figure}

\subsubsection*{Full Spectrum}
\label{sec:see-saw-spectrum}
Since the eigensystem of $\mb{B}(w,x)$ is known analytically
(Eqs.~\eqref{eq:B-evals}, \eqref{eq:B-evecs}),
the full spectrum and eigenstates of the Hamiltonian $\bar{\mb{H}_1}$
for the single clique $G_c$ are also known analytically.
In particular, the eigenvalues of the same-sign block \( \mb{B}(\mt{\weff_c}, \sqrt{n_c} \mt{x}) \) are:
\begin{equation}
\beta_k = -\tfrac{1}{2}\!\left(\mt{\weff_c} + (-1)^k \sqrt{[\mt{\weff_c}]^2 + n_c[\mt{x}]^2}\right), \quad k=0,1,
\label{eq:beta-value}
\end{equation}
with corresponding eigenvectors:
\begin{align}
\begin{cases}
\ket{\beta_0} = \tfrac{1}{\sqrt{1 + \gamma^2}} \left( \gamma \ket{0} + \ket{1} \right), \\
\ket{\beta_1} = \tfrac{1}{\sqrt{1 + \gamma^2}} \left( \ket{0} - \gamma \ket{1} \right),
\end{cases}
\label{eq:beta-vector}
\end{align}
with the mixing coefficient
\begin{align}
  \gamma = \tfrac{ \sqrt{n_c}\, \mt{x} }{ |\mt{\weff_c}| + \sqrt{ \left[\mt{\weff_c}\right]^2 + n_c \left[\mt{x}\right]^2 } }.
  \label{eq:gamma}
\end{align}
Note $\gamma$, $\beta_0$ and $\beta_1$ are all time-dependent, through their dependence on $\mt{x}$.

To summarize, the low-energy subspace of a single clique decomposes into 
a two-dimensional same-sign (spin-$\tfrac{1}{2}$) block and $n_c - 1$ 
one-dimensional opposite-sign (spin-0) blocks, with explicit operator 
representations in the $\Bc'$ basis. In the next subsection, we extend 
this construction to all cliques in $L$ to obtain the full block 
decomposition of $\mb{H}^{\ms{low}}$.

\subsubsection{Block Decomposition of $\mb{H}^{\ms{low}}$}
\label{sec:block-low}

Using the single-clique angular momentum basis from 
Section~\ref{sec:single-clique}, we tensor over all cliques in $L$ to 
form a global basis for $\mc{L}_{-}$ in which $\mb{H}^{\ms{low}}$ is 
block-diagonal, details are described in
Section 9.1 of~\cite{Choi-Beyond}.

The angular momentum basis is constructed locally within each clique
using projection and transformation operators (Section~\ref{sec:single-clique}),
and then combined into a global basis $\Ba$ by tensoring over all cliques in $L$
with the identity on $\mc{V}_R$.
This preserves the tensor-product structure and yields a block-diagonal form
separating the same-sign and opposite-sign sectors.

The decomposition emerges hierarchically:
each clique yields one same-sign sector (effective spin-$\tfrac{1}{2}$)
and several opposite-sign sectors (spin-0).
At the next level, the cliques in $L$---forming the \MIC{}---define a
bare subsystem with a same-sign sector $\mc{C}^{\ms{bare}}$ and
opposite-sign sectors $\mc{W}^{\ms{bare}}, \mc{Q}^{\ms{bare}}$.
Finally, the full same-sign and opposite-sign sectors are
\[
\mc{C} = \mc{C}^{\ms{bare}} \otimes \left( \mathbb{C}^2 \right)^{\otimes m_r},\quad
\mc{W} = \mc{W}^{\ms{bare}} \otimes \left( \mathbb{C}^2 \right)^{\otimes m_r},\quad
\mc{Q} = \mc{Q}^{\ms{bare}} \otimes \left( \mathbb{C}^2 \right)^{\otimes m_r},
\]
where $m_r = |R|$.
The full Hamiltonian then decomposes into $\mb{H}_{\mc{C}}$,
$\mb{H}_{\mc{Q}}$, and $\mb{H}_{\mc{W}}$, which are decoupled in the disjoint
case and coupled in the shared case.

\subsubsection*{Confinement in the $\Jxx = 0$ Case}
In the $\Jxx = 0$ case, the opposite-sign blocks remain at higher energy
and are dynamically irrelevant before the anti-crossing.
The evolving ground state is confined to the same-sign block
$\mc{C}$, so the anti-crossing occurs entirely within it.

As shown in Corollary 9.2 of~\cite{Choi-Beyond},
the same-sign block $\mb{H}_{\mc{C}}$ has the unified form
\[
\mb{H}_{\mc{C}} =
\sum_{i \in L} \mb{B}_i \left( \mt{\weff_i},\, \sqrt{n_i} \mt{x} \right)
+ \sum_{j \in R} \mb{B}_j \left( w_j,\, \mt{x} \right)
+ \Jzz \sum_{(i,j) \in L \times R} f_i^{\ms{C}} \shz{i} \shz{j},
\]
with effective coefficients
\[
\mt{\weff_i} = w_i - \tfrac{n_i - 1}{4} \mt{jxx}, \quad
f_i^{\ms{C}} = 
\begin{cases}
1 & \text{for } \Gdis, \\
\tfrac{n_i - 1}{n_i} & \text{for } \Gshare,
\end{cases} 
\]
As in the single-clique case, we keep the $\mt{jxx}$ term in order to 
preserve compatibility with the more general formulation used in the
main paper, while our later results here will specialize to 
$\Jxx = 0$.

This structure is identical in the disjoint and shared cases, so the
mechanism of tunneling-induced anti-crossing---and the resulting
exponentially small gap---is governed by the same effective Hamiltonian.
Hence, it is sufficient to analyze the disjoint case (\Gdis{}) when
establishing the limitation of TFQA.

\section{Inner Decomposition of the Same-sign Block}
\label{sec:inner}
We recall the two inner decompositions of the same-sign block \( \mb{H}_{\mc{C}} \) ($L$-inner and $R$-inner), as described in Section~9.2 of~\cite{Choi-Beyond}.
In this paper, we consider the case \( \Jxx = 0 \) and focus on the \( L \)-inner decomposition.
Under this decomposition, the relevant low-energy behavior is captured by a coupled pair of zero-indexed blocks.
In particular, we describe how the ground state localizes into the coupled block \( \Hcore \), consisting of \( \mb{H}_L^{(0)} \) and \( \mb{H}_R^{(0)} \), linked through the empty-set basis state.
This localized evolution sets the stage for the tunneling-induced anti-crossing analyzed in the next section.

\subsubsection{Two Block Decompositions of \( \mb{H}_{\mc{C}} \): \( L \)-Inner vs.~\( R \)-Inner}
Physically, the same-sign block Hamiltonian \( \mb{H}_{\mc{C}} \) is symmetric under interchange of the \( L \) and \( R \) subsystems: permuting the tensor factors leaves the spectrum unchanged.  
Mathematically, this corresponds to a permutation similarity transformation: reordering the basis to place \( L \) before \( R \), or vice versa, yields a matrix with identical eigenvalues.
Combinatorially, this symmetry allows the matrix to be organized into a two-layer block structure with either subsystem---\( L \) or \( R \)---serving as the ``inner'' block.  
That is, \( \mb{H}_{\mc{C}} \) can be expressed either as a collection of \( L \)-blocks indexed by the states from \( R \), or as \( R \)-blocks indexed by the states from \( L \).  

For illustration, we assume uniform clique size \( n_i = n_c \) for all \( i \).  
As in the main paper, we restrict the same-sign block Hamiltonian to the symmetric subspace:
\begin{align*}
\mb{H}_{\mc{C}}^{\ms{sym}}
&= \mb{H}_L^{\ms{bare}} \otimes \mb{I}_R + \mb{I}_L \otimes \mb{H}_R^{\ms{bare}} + \mb{H}_{LR},
\end{align*}
where
\begin{align*}
\mb{H}_L^{\ms{bare}} &= - \sqrt{n_c}\, \mt{x}\, \mb{C}\Sop{\X}(m) - \weff\, \mb{C}\Sop{\sZ}(m), \\
\mb{H}_R^{\ms{bare}} &= - \mt{x}(t)\, \mb{C}\Sop{\X}(m_r) - w\, \mb{C}\Sop{\sZ}(m_r), \\
\mb{H}_{LR} &= \Jzz\, \mb{C}\Sop{\sZ}(m) \cdot \mb{C}\Sop{\sZ}(m_r).
\end{align*}
Here
\begin{align*}
\mb{C}\Sop{\sZ}(m) = 
\begin{bmatrix}
m & 0 & \cdots & 0 \\
0 & m - 1 & \cdots & 0 \\
\vdots & \vdots & \ddots & \vdots \\
0 & 0 & \cdots & 0\\
\end{bmatrix},
\quad
\label{eq:CSX}
\mb{C}\Sop{\X}(m) =
\begin{bmatrix}
0 & \tfrac{\sqrt{m}}{2} & 0 & \cdots & 0 \\
\tfrac{\sqrt{m}}{2} & 0 & \tfrac{\sqrt{2(m - 1)}}{2} & \cdots & 0 \\
0 & \tfrac{\sqrt{2(m - 1)}}{2} & 0 & \cdots & 0 \\
\vdots & \vdots & \vdots & \ddots & \tfrac{\sqrt{m}}{2} \\
0 & 0 & 0 & \tfrac{\sqrt{m}}{2} & 0
\end{bmatrix}.
\end{align*}
are the collective Pauli operators restricted to the symmetric
subspace of \( m \) spins, and all \(^{\ms{bare}}\) operators are understood in this context to be restricted to the symmetric subspace.


To illustrate the resulting block structure, we consider a concrete example with \( m = 2 \) and \( m_r = 3 \).
\begin{figure}[!htbp]
  \centering
  \includegraphics[width=0.6\textwidth]{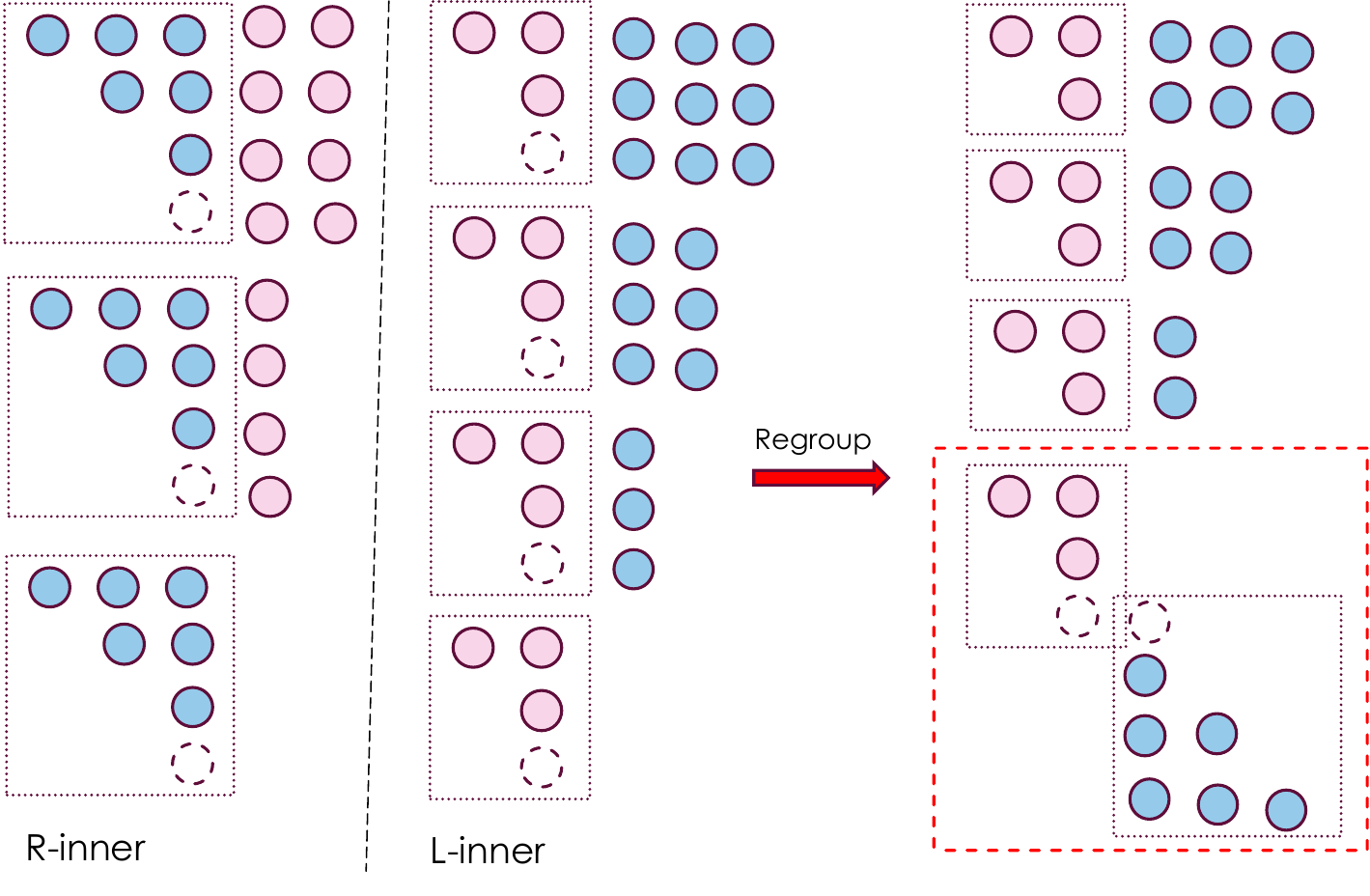}
\caption{
  Two possible orderings of the basis states in the symmetric subspace of the same-sign block \( \mb{H}_{\mc{C}}^{\ms{sym}} \), illustrated for \( m = 2 \), \( m_r = 3 \), corresponding to the \( R \)-inner and \( L \)-inner decompositions.  
  The \( R \) subsystem is shown in blue, and \( L \) in pink. The dashed circle marks the empty-set basis state (no spin-ups).  
  In the \( L \)-inner representation (right), the basis states are reordered so that the zero-indexed \( R \)-block appears in the bottom-right corner.  
  The red dashed box highlights the core block \( \Hcore \), consisting of the blocks \( \mb{H}_L^{(0)} \) and \( \mb{H}_R^{(0)} \), coupled through the empty-set state.
}
\label{fig:L-R-inner}
\end{figure}

\paragraph{Explicit Matrix Representation of \( L \)-inner Block Decomposition.}
We present the explicit matrix form of \( \mb{H}_{\mc{C}}^{\ms{sym}} \),  
corresponding to the decomposition shown in Figure~\ref{fig:L-R-inner}.  
Each diagonal entry reflects the total energy of a basis state,  
while off-diagonal entries arise from the transverse-field term, which connects basis states differing by a single spin flip.

There are \( m_r + 1 = 4 \) outer-layer blocks, each a \( (m+1) \times (m+1) \) matrix acting on the \( L \) subsystem.  
The full matrix representation is denoted \( \mb{H}_{\mc{C}}^{L\text{-}\ms{inner}} \):
\begin{center}
\scalebox{0.5}{$
\left(
\begin{array}{ccc|ccc|ccc|ccc}
 6 \Jzz - 3w - 2\weff & -\tfrac{\sqrt{n_c}\, x}{\sqrt{2}} & 0 & -\tfrac{1}{2} \sqrt{3}x & 0 & 0 & 0 & 0 & 0 & 0 & 0 & 0 \\
 -\tfrac{\sqrt{n_c}\, x}{\sqrt{2}} & 3 \Jzz - 3w - \weff & -\tfrac{\sqrt{n_c}\, x}{\sqrt{2}} & 0 & -\tfrac{1}{2} \sqrt{3}x & 0 & 0 & 0 & 0 & 0 & 0 & 0 \\
 0 & -\tfrac{\sqrt{n_c}\, x}{\sqrt{2}} & -3w & 0 & 0 & -\tfrac{1}{2} \sqrt{3}x & 0 & 0 & 0 & 0 & 0 & 0 \\
 \hline
 -\tfrac{1}{2} \sqrt{3}x & 0 & 0 & 4 \Jzz - 2w - 2\weff & -\tfrac{\sqrt{n_c}\, x}{\sqrt{2}} & 0 & -x & 0 & 0 & 0 & 0 & 0 \\
 0 & -\tfrac{1}{2} \sqrt{3}x & 0 & -\tfrac{\sqrt{n_c}\, x}{\sqrt{2}} & 2 \Jzz - 2w - \weff & -\tfrac{\sqrt{n_c}\, x}{\sqrt{2}} & 0 & -x & 0 & 0 & 0 & 0 \\
 0 & 0 & -\tfrac{1}{2} \sqrt{3}x & 0 & -\tfrac{\sqrt{n_c}\, x}{\sqrt{2}} & -2w & 0 & 0 & -x & 0 & 0 & 0 \\
 \hline
 0 & 0 & 0 & -x & 0 & 0 & 2 \Jzz - w - 2\weff & -\tfrac{\sqrt{n_c}\, x}{\sqrt{2}} & 0 & -\tfrac{1}{2} \sqrt{3}x & 0 & 0 \\
 0 & 0 & 0 & 0 & -x & 0 & -\tfrac{\sqrt{n_c}\, x}{\sqrt{2}} & \Jzz - w - \weff & -\tfrac{\sqrt{n_c}\, x}{\sqrt{2}} & 0 & -\tfrac{1}{2} \sqrt{3}x & 0 \\
 0 & 0 & 0 & 0 & 0 & -x & 0 & -\tfrac{\sqrt{n_c}\, x}{\sqrt{2}} & -w & 0 & 0 & -\tfrac{1}{2} \sqrt{3}x \\
 \hline
 0 & 0 & 0 & 0 & 0 & 0 & -\tfrac{1}{2} \sqrt{3}x & 0 & 0 & -2\weff & -\tfrac{\sqrt{n_c}\, x}{\sqrt{2}} & 0 \\
 0 & 0 & 0 & 0 & 0 & 0 & 0 & -\tfrac{1}{2} \sqrt{3}x & 0 & -\tfrac{\sqrt{n_c}\, x}{\sqrt{2}} & -\weff & -\tfrac{\sqrt{n_c}\, x}{\sqrt{2}} \\
 0 & 0 & 0 & 0 & 0 & 0 & 0 & 0 & -\tfrac{1}{2} \sqrt{3}x & 0 & -\tfrac{\sqrt{n_c}\, x}{\sqrt{2}} & 0 \\
\end{array}
\right)
$}
\end{center}

\subsubsection{Localization and Effective Hamiltonian \( \Hcore \)}
Assume \( \Gamma_1 \) is sufficiently large so that the initial ground state is the uniform superposition.  
As the parameter \( \mt{x} \) decreases, the \( L \)-spins experience a stronger effective transverse field due to the extra \( \sqrt{n_c} \) factor.  
Since this transverse field remains much larger than the \( z \)-field, the latter can be neglected to leading order.  
Under this approximation, the amplitudes of basis states within the block \( \mb{H}_R^{(0)} \) become suppressed---they are negligibly small compared to those in \( \mb{H}_L^{(0)} \), except for their coupling through the shared basis (empty-set) state.

This motivates us to reorder the matrix representation so that \( \mb{H}_R^{(0)} \) appears explicitly in the bottom-right corner. To extract the block \( \mb{H}_R^{(0)} \) from the \( L \)-inner representation, we permute the rows and columns of \( H_{\mc{C}}^{\ms{L-inner}} \) so that the indices corresponding to the \( i_R = 0 \) sector appear contiguously at the end.  
This reordering moves \( \mb{H}_R^{(0)} \) into the bottom-right corner of the matrix.
Note that the original blocks \( \mb{H}_L^{(0)} \) and \( \mb{H}_R^{(0)} \) are coupled through the basis state corresponding to the empty set,  
which is shared between them. See Figure~\ref{fig:L-R-inner} (right panel) for illustration.
The explicit form of the low-energy submatrix,
with the shared row and column shaded, is:
\[
\begin{bmatrix}
 -2\weff & -\tfrac{\sqrt{n_c}\, x}{\sqrt{2}} & \cellcolor{gray!20} 0 & 0 & 0 & 0 \\
 -\tfrac{\sqrt{n_c}\, x}{\sqrt{2}} & -\weff & \cellcolor{gray!20} -\tfrac{\sqrt{n_c}\, x}{\sqrt{2}} & 0 & 0 & 0 \\
 \cellcolor{gray!20} 0 & \cellcolor{gray!20} -\tfrac{\sqrt{n_c}\, x}{\sqrt{2}} & \cellcolor{gray!20} 0 & \cellcolor{gray!20} -\tfrac{1}{2} \sqrt{3}\, x & \cellcolor{gray!20} 0 & \cellcolor{gray!20} 0 \\
 0 & 0 & \cellcolor{gray!20} -\tfrac{1}{2} \sqrt{3}\, x & -w & -x & 0 \\
 0 & 0 & \cellcolor{gray!20} 0 & -x & -2w & -\tfrac{1}{2} \sqrt{3}\, x \\
 0 & 0 & \cellcolor{gray!20} 0 & 0 & -\tfrac{1}{2} \sqrt{3}\, x & -3w
\end{bmatrix}
\]

The evolving ground state localizes into the lowest-energy block \( \mb{H}_L^{(0)} \) before the anti-crossing,
as shown in Figure~\ref{fig:block-proj-Jxx0}.
\begin{figure}[!htbp]
  \centering
  \includegraphics[width=0.65\textwidth]{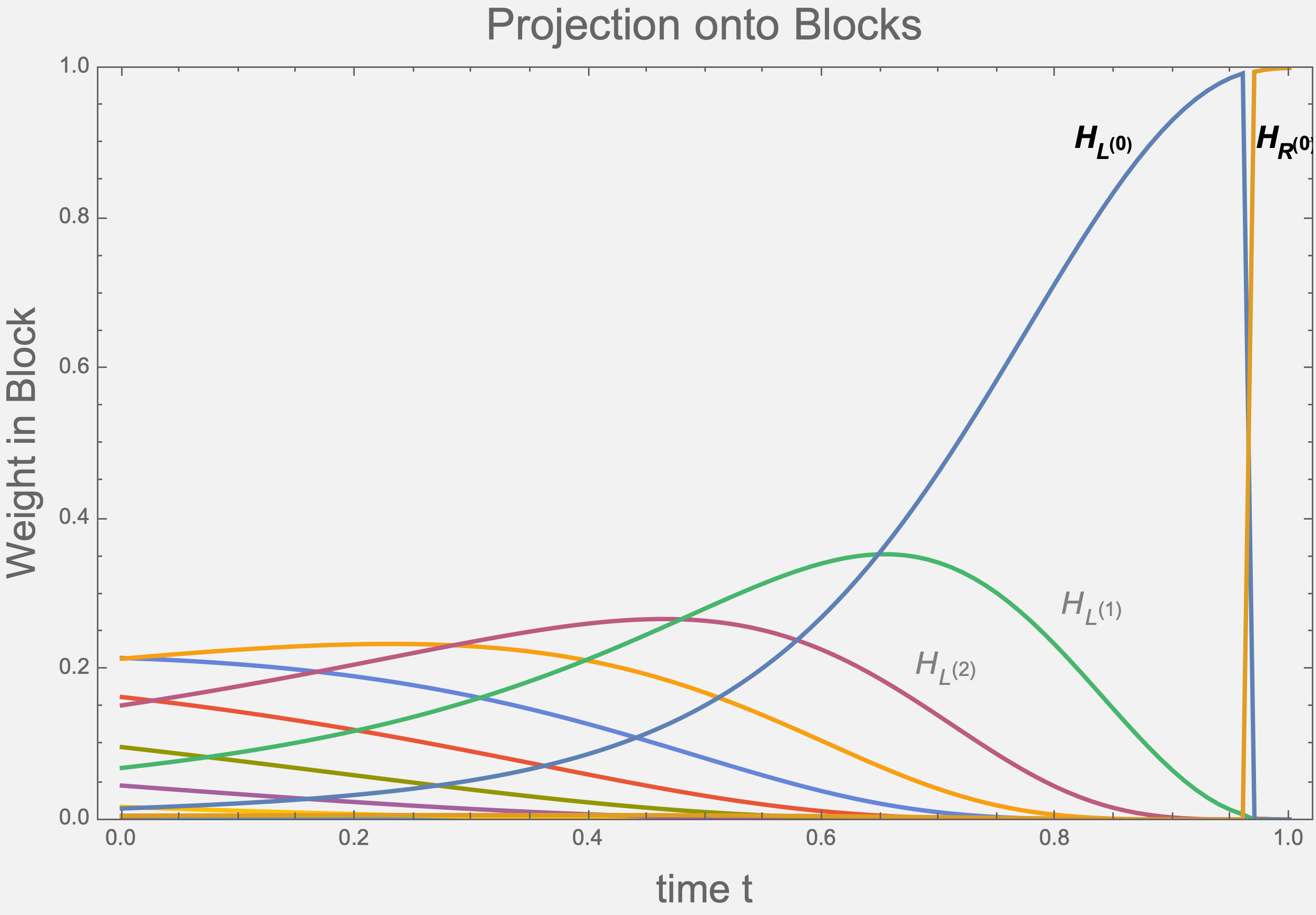}
  \caption{
    Projection of the evolving ground state onto each block of the same-sign Hamiltonian \(\mb{H}_{\mc{C}}\), in the disjoint case with \(\Jxx = 0\).  
    The ground state begins near-uniform, then steers smoothly into the \( \mb{H}_L^{(0)} \) block, which dominates before the anti-crossing.  
    An anti-crossing occurs near \( 0.95 \), where amplitude tunnels from \( \mb{H}_L^{(0)} \) to \( \mb{H}_R^{(0)} \).  
    Intermediate blocks \( \mb{H}_L^{(1)} \) and \( \mb{H}_L^{(2)} \) participate during structural compression but remain subdominant.
  }
  \label{fig:block-proj-Jxx0}
\end{figure}
This localization justifies reducing the full same-sign block Hamiltonian \( \mb{H}_{\mc{C}} \) to an effective Hamiltonian  
\[
\Hcore = \mb{H}_L^{(0)} + \mb{H}_R^{(0)}.
\]
where
\(\mb{H}_L^{(0)} = \mb{H}_L^{\ms{bare}} \otimes \ket{0}_R\bra{0},
\mb{H}_R^{(0)} = \ket{0}_L\bra{0} \otimes \mb{H}_R^{\ms{bare}}.
\)
 Figure~\ref{fig:energy-HC-vs-Hcore} shows that the lowest energy levels of \( \Hcore \) closely match those of the full same-sign Hamiltonian \( \mb{H}_{\mc{C}} \).

\begin{figure}[!htbp]
  \centering
  $$
  \begin{array}{cc}
    \includegraphics[width=0.45\textwidth]{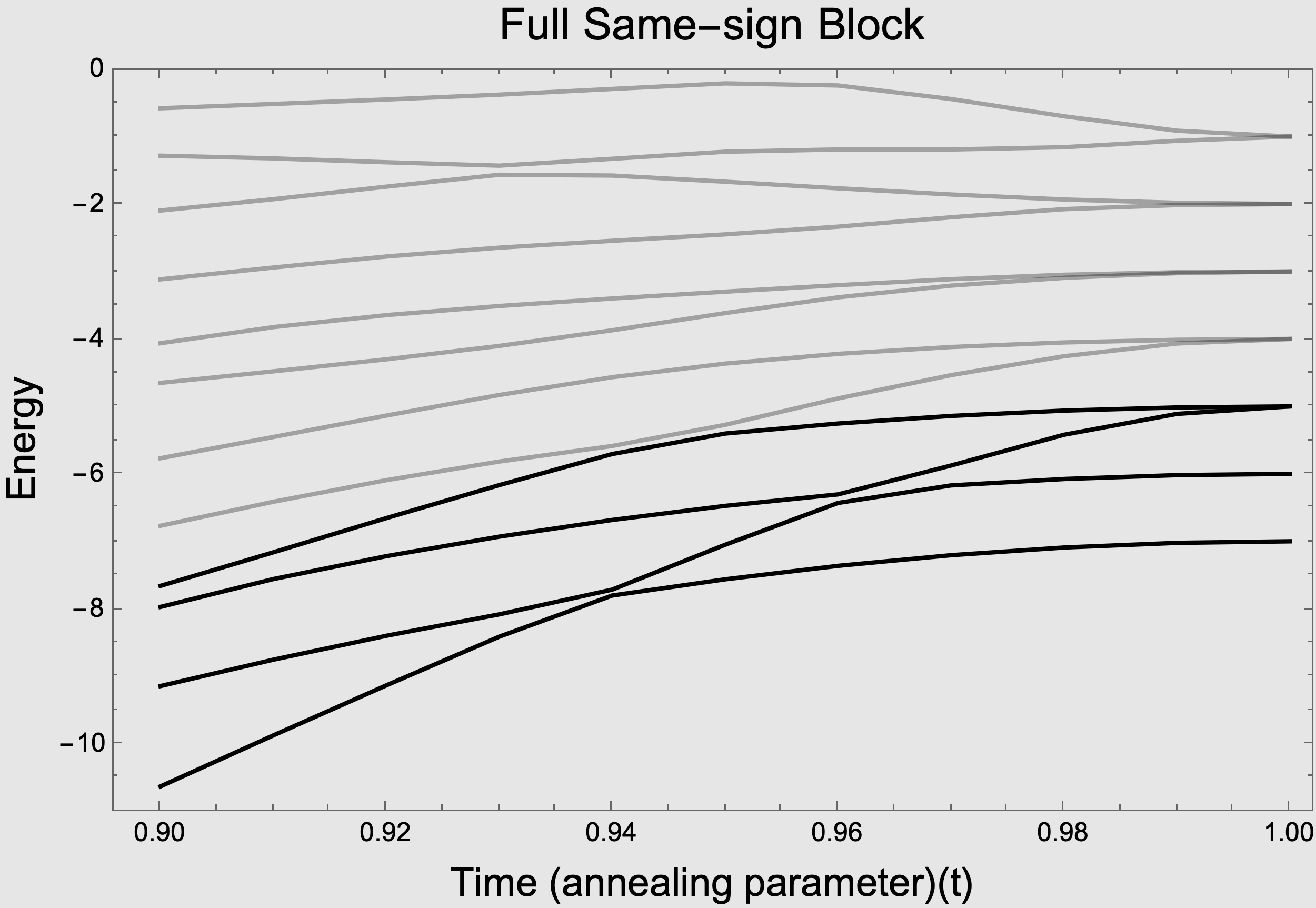} &
    \includegraphics[width=0.45\textwidth]{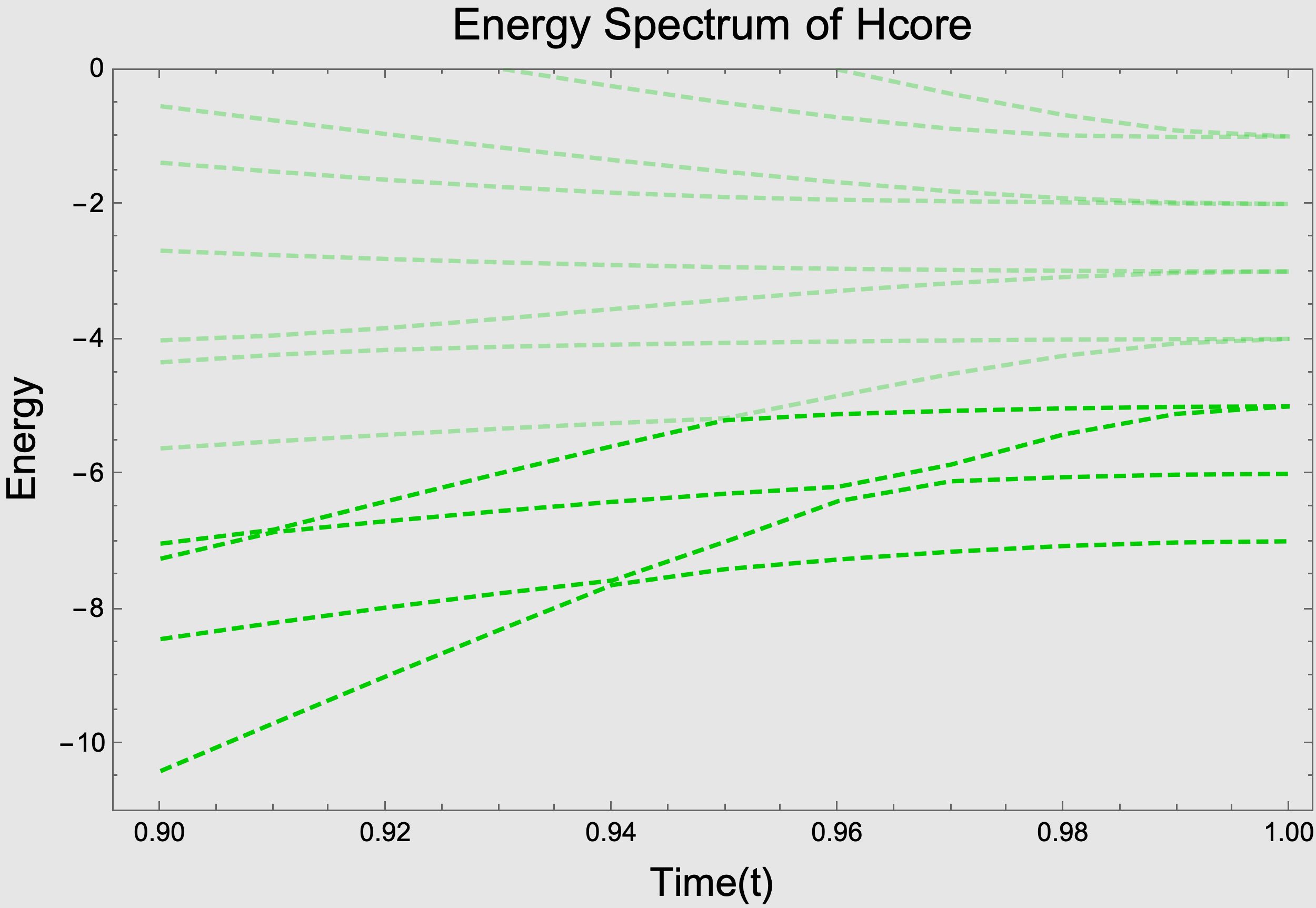}
  \end{array}
  $$
  \includegraphics[width=0.65\textwidth]{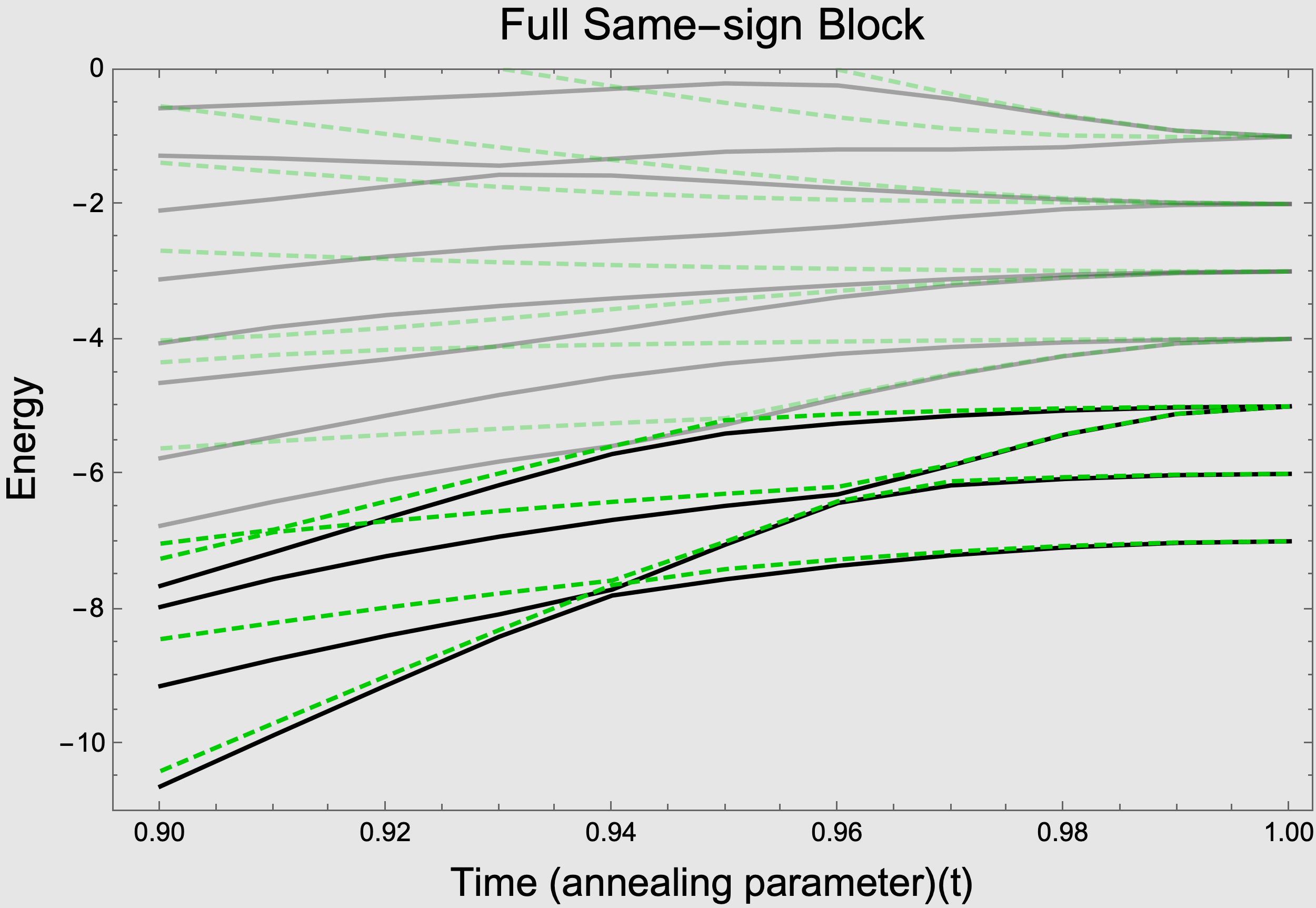}
  \caption{
    Comparison of the energy spectrum of the full same-sign block Hamiltonian \( \mb{H}_{\mc{C}} \)  
    and the reduced two-block Hamiltonian \( \Hcore \), in the disjoint case with \( \Jxx = 0 \).  
    The top-left panel shows the full spectrum of \( \mb{H}_{\mc{C}} \) (solid black and gray),  
    and the top-right panel shows the spectrum of \( \Hcore \) (green dashed).  
    The bottom panel overlays the two, demonstrating that the lowest energy levels of \( \Hcore \) closely track those of \( \mb{H}_{\mc{C}} \), with the anti-crossing position shifted slightly to the right.
  }
  \label{fig:energy-HC-vs-Hcore}
\end{figure}

\section{Detailed Analysis of \( \Hcore \)}
\label{sec:Hcore}

In this section, we analyze the spectrum of \( \Hcore \) using the exact eigensystems of the embedded bare subsystems \( \mb{H}_L^{(0)} \) and \( \mb{H}_R^{(0)} \), each of which is analytically solvable.  
The analysis proceeds by reformulating the eigenvalue problem as a generalized eigenvalue problem, using a non-orthogonal basis constructed by extending the bare eigenstates of the subsystems to the full Hilbert space.

This transformation enables a clean perturbative treatment and forms the foundation for a structured spectral analysis. The approach proceeds in three main steps:
\begin{itemize}
  \item Describe the bare eigenstates of \( \Hcore \) (Section~\ref{sec:core-bare}).

  \item Reformulate the eigenvalue problem for \( \Hcore \) as a generalized eigenvalue problem in the non-orthogonal basis (Section~\ref{sec:reformulation}).

  \item Develop the perturbative structure of the generalized eigenvalue system (Section~\ref{sec:pert-generalized}), including:
  \begin{enumerate}
    \item Showing that the true energy values (away from the anti-crossing) are well-approximated by the bare energies (Section~\ref{sec:approx-pert}).
    
    \item Analyzing the anti-crossing structure by approximating the lowest two energy levels, both away from and near the crossing,
      and deriving a perturbative bound for the anti-crossing gap using an effective \( 2 \times 2 \) Hamiltonian (Section~\ref{sec:ac-gap}).
  \end{enumerate}
\end{itemize}

\begin{remark}
  The final step, constructing a reduced effective Hamiltonian for the generalized system, is structurally similar to the method of effective Hamiltonians for the non-orthogonal basis developed in~\cite{AndradeFreire2003}. However, in their setting, the non-orthogonal basis is derived from the physical structure of the original problem (e.g., localized atomic orbitals).
  In contrast, our non-orthogonal basis is synthetically constructed from bare eigenstates and does not correspond to the physical structure of the full system.
\end{remark}

\subsection{Bare Eigenstates of \( \Hcore \)}
\label{sec:core-bare}
In this section, we describe the bare eigenstates of \( \Hcore \) by relating the embedded bare operators \( \mb{H}_L^{(0)} \), \( \mb{H}_R^{(0)} \) to the bare subsystem Hamiltonians \( \mb{H}_L^{\ms{bare}} \), \( \mb{H}_R^{\ms{bare}} \).
Recall that the effective Hamiltonian \( \Hcore \) is
\[
\Hcore = \mb{H}_L^{(0)} + \mb{H}_R^{(0)},
\]
where
\[\mb{H}_L^{(0)} = \mb{H}_L^{\ms{bare}} \otimes \ket{0}_R\bra{0},\quad
\mb{H}_R^{(0)} = \ket{0}_L\bra{0} \otimes \mb{H}_R^{\ms{bare}}.
\]
The eigenvalues and eigenstates of the the bare Hamiltonians \( \mb{H}_L^{\ms{bare}} \), \( \mb{H}_R^{\ms{bare}} \) are
\[
\HL \ket{L_i} = e_{L_i} \ket{L_i}, \quad 
\HR \ket{R_j} = e_{R_j} \ket{R_j},
\]
where \( \{ \ket{L_i} \} \) and \( \{ \ket{R_j} \} \) are analytically known (Section~7.2 in \cite{Choi-Beyond})
and form orthonormal eigenbases of the \( L \)- and \( R \)-subspaces.


The corresponding padded eigenstates in the full space are
\[
\ket{\bar{L}_i} = \ket{L_i} \otimes \ket{\emptyset}_R, \quad 
\ket{\bar{R}_j} = \ket{\emptyset}_L \otimes \ket{R_j},
\]
and they satisfy
\[
\mb{H}_L^{(0)} \ket{\bar{L}_i} = e_{L_i} \ket{\bar{L}_i}, \quad 
\mb{H}_R^{(0)} \ket{\bar{R}_j} = e_{R_j} \ket{\bar{R}_j}.
\]

\begin{remark}
We use a bar to distinguish padded eigenstates of the full Hilbert space from bare eigenstates of the subsystem Hamiltonians.
For instance, \( \ket{L_i} \) is an eigenstate of the bare Hamiltonian \( \HL \) acting on the \( L \)-subspace, while \( \ket{\bar{L}_i} = \ket{L_i} \otimes \ket{\emptyset}_R \) is the corresponding eigenstate of the embedded operator \( \mb{H}_L^{(0)} \) acting on the full Hilbert space. The same applies to \( \ket{R_j} \) and \( \ket{\bar{R}_j} \).
\end{remark}

Recall that the full eigensystem for both \( \mb{H}_L \) and \( \mb{H}_R \) (and thus for \( \mb{H}_L^{(0)} \) and \( \mb{H}_R^{(0)} \)) is analytically known.
In the following, we first derive the results in terms of the bare eigenstates \( \ket{\bar{L}_i} \) and \( \ket{\bar{R}_j} \), and then substitute their explicit forms in the anti-crossing structure section later.

\subsection{Generalized Eigenvalue Reformulation}
\label{sec:reformulation}

We express the Hamiltonian \( \Hcore \) in the basis of padded bare eigenstates
\(
\left\{ \ket{\bar{L}_i} \right\} \cup \left\{ \ket{\bar{R}_j} \right\}
\), which leads to a generalized eigenvalue problem.
We define
\[
\Hgen \mdef
\begin{bmatrix}
\mb{H}_{RR} & \mb{H}_{RL} \\
\mb{H}_{LR} & \mb{H}_{LL}
\end{bmatrix}, \quad
\Sgen \mdef
\begin{bmatrix}
\mb{S}_{RR} & \mb{S}_{RL} \\
\mb{S}_{LR} & \mb{S}_{LL}
\end{bmatrix},
\]
where \( \mb{H}_{RR} \) contains the entries \( \bra{\bar{R}_i} \Hcore \ket{\bar{R}_j} \), and similarly for \( \mb{H}_{LL} \), \( \mb{H}_{LR} \), and \( \mb{H}_{RL} \).  
The matrix \( \mb{S}_{RR} \) contains the overlaps \( \langle \bar{R}_i \vert \bar{R}_j \rangle \), and similarly for the remaining blocks of \( \Sgen \).

We now consider the generalized eigenvalue equation:
\begin{align}
\Hgen \ket{\Egen_n} = \Egen_n \, \Sgen \ket{\Egen_n},
\label{eq:g1}
\end{align}
where the pair \( (\Egen_n, \ket{\Egen_n}) \) is referred to as a \emph{generalized eigenpair} of the generalized eigenvalue system \( (\Hgen, \Sgen) \).  
We refer to \( \Hgen \) as the \emph{generalized Hamiltonian}, and \( \Sgen \) as the corresponding \emph{overlap matrix}.
In the numerical linear algebra literature, such pairs are also called \emph{matrix pencils}.\footnote{The term “matrix pencil” refers to the parametrized family \( \Hgen - \lambda \Sgen \), whose roots \( \lambda = \Egen_n \) yield the eigenvalues of the system.}

\begin{theorem}
\label{thm:generalized-eigenpair-equivalence}
Let \( (E_n, \ket{E_n}) \) be an eigenpair of the original Hamiltonian \( \Hcore \), and let \( (\Egen_n, \ket{\Egen_n}) \) be a generalized eigenpair of \( (\Hgen, \Sgen) \) as in~\eqref{eq:g1}.
Then:
\begin{itemize}
  \item Every eigenvalue \( E_n \) of \( \Hcore \) appears as a generalized eigenvalue \( \Egen_n \).
  \item Conversely, every generalized eigenvalue \( \Egen_n \) corresponds to an eigenvalue \( E_n \) of \( \Hcore \).
\end{itemize}
The eigenvalues coincide: \( \Egen_n = E_n \) for all \( n \), while the eigenvectors \( \ket{\Egen_n} \) and \( \ket{E_n} \) differ.
\end{theorem}

\begin{proof}
We can write the eigenstate \( \ket{E_n} \) of \( \Hcore \) in the padded bare basis as
\[
\ket{E_n} = \sum_{k=0}^{2^{m_R}-1} c^{\ms{R}}_k \ket{\bar{R}_k} + \sum_{k=0}^{2^{m_L}-1} c^{\ms{L}}_k \ket{\bar{L}_k}.
\]
Note that the coefficients \( c^{\ms{R}} \) and \( c^{\ms{L}} \) may not be unique, since the set \( \{ \ket{\bar{R}_k}, \ket{\bar{L}_k} \} \) spans a space of dimension \( 2^{m_L} + 2^{m_R} \), whereas the true Hilbert space of \( \Hcore \) has dimension \( 2^{m_L} + 2^{m_R} - 1 \) due to the identification \( \ket{\bar{R}_0} = \ket{\bar{L}_0} \). Thus, the padded set is linearly dependent.

Now rewrite the eigenvalue equation:
\[
\Hcore \left( \sum_k c^{\ms{R}}_k \ket{\bar{R}_k} + \sum_k c^{\ms{L}}_k \ket{\bar{L}_k} \right)
= E_n \left( \sum_k c^{\ms{R}}_k \ket{\bar{R}_k} + \sum_k c^{\ms{L}}_k \ket{\bar{L}_k} \right).
\]

Taking the inner product with \( \bra{\bar{R}_i} \) and \( \bra{\bar{L}_i} \), we obtain the generalized eigenvalue equation in matrix form:
\[
\begin{bmatrix}
\mb{H}_{RR} & \mb{H}_{RL} \\
\mb{H}_{LR} & \mb{H}_{LL}
\end{bmatrix}
\begin{bmatrix}
\vec{c}_{\ms{R}} \\
\vec{c}_{\ms{L}}
\end{bmatrix}
= E_n
\begin{bmatrix}
\mb{S}_{RR} & \mb{S}_{RL} \\
\mb{S}_{LR} & \mb{S}_{LL}
\end{bmatrix}
\begin{bmatrix}
\vec{c}_{\ms{R}} \\
\vec{c}_{\ms{L}}
\end{bmatrix}.
\]

Then we have
\(
\Hgen \ket{\Egen_n} = \Egen_n \Sgen \ket{\Egen_n},
\)
with \( \Egen_n = E_n \) and \( \ket{\Egen_n} = \begin{bmatrix} \vec{c}_{\ms{R}} \\ \vec{c}_{\ms{L}} \end{bmatrix} \).

This shows that every eigenpair \( (E_n, \ket{E_n}) \) of \( \Hcore \) corresponds to a generalized eigenpair \( (\Egen_n, \ket{\Egen_n}) \) of \( (\Hgen, \Sgen) \), though \( \ket{\Egen_n} \) may not be unique.
Conversely, each generalized eigenpair \( (\Egen_n, \ket{\Egen_n}) \) of \( (\Hgen, \Sgen) \) determines an eigenpair \( (E_n, \ket{E_n}) \) of \( \Hcore \) with \( \Egen_n = E_n \).
\end{proof}



\begin{center}
\fbox{
\begin{minipage}{0.92\textwidth}
\textbf{Summary of the Reformulation.}
We express the eigenvalue problem for \( \Hcore \) in a non-orthogonal (normalized) basis:
\[
\big\{ \ket{\bar{L}_i} \big\} \cup \big\{ \ket{\bar{R}_j} \big\}.
\]
Each subset is orthonormal, but the two sets are not mutually orthogonal.

This yields a generalized eigenvalue problem of the form
\[
\Hgen \ket{\Egen_n} = \Egen_n \, \Sgen \ket{\Egen_n}.
\]

This reformulation preserves the spectrum:
\[
\boxed{ \Egen_n = E_n \quad \text{for all } n, }
\]
allowing us to analyze the dynamics using the generalized eigenvalue system \( (\Hgen, \Sgen) \).
This forms the foundation for our perturbative analysis in the next subsection.
\end{minipage}
}
\end{center}

\subsection{Perturbative Structure of the Generalized Eigenvalue System}
\label{sec:pert-generalized}

In this section, we develop the perturbative structure of the generalized eigenvalue system \( (\Hgen, \Sgen) \), and use it to analyze the spectrum of \( \Hcore \). Specifically:
\begin{enumerate}
  \item In Section~\ref{sec:approx-pert}, we show that the true energy values of \( \Hcore \) (away from the anti-crossing) are well-approximated by the bare energies.
  \item In Section~\ref{sec:ac-gap}, we derive a \( 2 \times 2 \) effective Hamiltonian that captures the tunneling-induced anti-crossing and provides a perturbative bound on the gap size.
\end{enumerate}

We decompose the generalized Hamiltonian and overlap matrix as
\[
\Hgen = \mb{H} + \YH, \quad \Sgen = \mb{I} + \DI,
\]
where
\begin{align*}
\mb{H} &=
\begin{bmatrix}
\mb{H}_{RR} & 0 \\
0 & \mb{H}_{LL}
\end{bmatrix}
\quad \text{(block-diagonal)}, &
\YH &=
\begin{bmatrix}
0 & \mb{H}_{RL} \\
\mb{H}_{LR} & 0
\end{bmatrix}
\quad \text{(off-diagonal)}, \\[1ex]
\mb{I} &=
\begin{bmatrix}
\mb{I}_{RR} & 0 \\
0 & \mb{I}_{LL}
\end{bmatrix}, &
\DI &=
\begin{bmatrix}
0 & \mb{S}_{RL} \\
\mb{S}_{LR} & 0
\end{bmatrix}.
\end{align*}

We now introduce a parameterized family of matrices:
\[
\mb{H}(\lambda) \mdef \mb{H} + \lambda \YH, \quad
\mb{I}(\lambda) \mdef \mb{I} + \lambda \DI, \quad \lambda \in [0,1),
\]
and consider the associated generalized eigenvalue problem:
\[
\mb{H}(\lambda) \ket{e_n(\lambda)} = e_n(\lambda) \mb{I}(\lambda) \ket{e_n(\lambda)}.
\]

At \( \lambda = 0 \), the eigenvalues \( e_n(0) \) correspond to the bare subsystem energies.  
At \( \lambda = 1 \), the eigenvalues \( e_n(1) = E_n \) recover the true eigenvalues of \( \Hcore \).

We analyze how these eigenvalues deform as a function of \( \lambda \).

\begin{remark}
There is one more eigenvalue in the unperturbed generalized eigenvalue system \( (\Hgen(0), \Sgen(0)) \) than in the fully coupled system \( (\Hgen(1), \Sgen(1)) \).  
While perturbations typically split degenerate energy levels, here the effect is reversed: the perturbation merges two distinct eigenvalues into a single degenerate level.  
In particular, the zero-energy levels from both subsystems appear to collapse into one.
\end{remark}

Although \( \Sgen = \mb{I} + \DI \) is not invertible, the matrix \( \mb{I} + \lambda \DI \) is invertible for all \( \lambda \in [0,1) \), as shown in the lemma below.

\begin{lemma}
\label{lemma:invertible}
The matrix \( \mb{I} + \DI \) is not invertible. However, for all \( \lambda \in [0,1) \), the matrix \( \mb{I} + \lambda \DI \) is invertible.
\end{lemma}

\begin{proof}
The overlap matrix \( \Sgen = \mb{I} + \DI \) is defined using the padded bare basis \( \{ \ket{\bar{L}_k}, \ket{\bar{R}_k} \} \), where \( \ket{\bar{L}_0} = \ket{\bar{R}_0} \).  
This identification introduces a linear dependence in the set, so \( \Sgen \) is singular and not invertible.

In contrast, for any \( \lambda \in [0,1) \), the matrix \( \mb{I} + \lambda \DI \) is a Hermitian perturbation of the identity by an operator of strictly subunit norm.  
All eigenvalues of \( \lambda \DI \) lie strictly within \( (-1, 1) \), so the eigenvalues of \( \mb{I} + \lambda \DI \) are strictly positive.  
Hence, \( \mb{I} + \lambda \DI \) is invertible.
\end{proof}


Using Lemma~\ref{lemma:invertible}, we can now convert the generalized eigenvalue
problem into a standard eigenvalue problem.
For \( \lambda \in [0,1) \), we rewrite the generalized eigenvalue problem
\[
  \left( \mb{H} + \lambda \YH \right) \ket{e_n(\lambda)} 
  = e_n(\lambda) \left( \mb{I} + \lambda \DI \right) \ket{e_n(\lambda)},
\]
in standard form using the transformation
\(
 \mb{T}(\lambda) \mdef \left( \mb{I} + \lambda \DI \right)^{-1/2}.
\)
This yields the equivalent standard eigenvalue equation:
\[
\mb{H}_{\ms{QI}}(\lambda)\ket{e_n(\lambda)} 
  = e_n(\lambda) \ket{e_n(\lambda)}.
\]
where $\mb{H}_{\ms{QI}}(\lambda) := \mb{T}(\lambda) \left( \mb{H} + \lambda \YH \right)\mb{T}(\lambda)$.
We refer to \( \mb{H}_{\ms{QI}}(\lambda) \) as the \emph{quasi-interpolated Hamiltonian}.
At \( \lambda = 0 \), the unperturbed equation becomes
\(
  \mb{H} \ket{e_n} = e_n \ket{e_n},
\)
where \( e_n = e_n(0) \) denotes the bare eigenvalue.

\subsubsection{Perturbative Approximation Scheme for the True Energies}
\label{sec:approx-pert}

In this section, we develop a perturbative scheme for approximating the difference 
\( e_n(1) - e_n \), where \( e_n(1) = E_n \) denotes the true eigenvalue of \( \Hcore \), 
and \( e_n = e_n(0) \) is the corresponding bare eigenvalue.
We begin by showing in Theorem~\ref{thm:2.1} that for any \( \lambda \in (0,1) \), 
the eigenvalue deformation \( e_n(\lambda) - e_n \) is exactly captured by a quadratic form involving the transformation matrix \( \mb{T}(\lambda) \) and the coupling term \( \YH - e_n \DI \).
Then, in Proposition~\ref{cor:lambda-small}, we show that when the level \( e_n \) is non-degenerate 
(i.e., well-separated from other bare levels), the eigenvalue deformation admits a second-order perturbative expansion in \( \lambda \).
Finally, we approximate the true energy at \( \lambda = 1 \) using the same expansion, as given in Corollary~\ref{cor:lambda-one}.

\begin{mdframed}
\begin{theorem}
  \label{thm:2.1}
  For \( 0 < \lambda < 1 \), let \( e_n = e_n(0) \) denote the unperturbed eigenvalue, and \( \ket{e_n^{(0)}} \) its corresponding eigenvector.  
  Then the perturbed eigenvalue \( e_n(\lambda) \) satisfies
  \[
    e_n(\lambda) - e_n 
    = \lambda \bra{e_n^{(0)}} \mb{T}(\lambda) \left( \YH - e_n \DI \right) \mb{T}(\lambda) \ket{e_n(\lambda)}.
  \]
\end{theorem}
\end{mdframed}

\begin{proof}
First, add and subtract an auxiliary term \( \lambda e_n \DI \) to \( \mb{H}(\lambda) \):
\begin{align*}
  \mb{H}(\lambda) = \mb{H} + \lambda \YH 
  = \left( \mb{H} + \lambda e_n \DI \right) + \lambda \left( \YH - e_n \DI \right).
\end{align*}

We now express $\mb{H}_{\ms{QI}}(\lambda)$ in terms of the unperturbed Hamiltonian:
\begin{align}
 \mb{H}_{\ms{QI}}(\lambda)
  =\mb{T}(\lambda) \left( \mb{H} + \lambda e_n \DI \right)\mb{T}(\lambda) 
  + \lambda\mb{T}(\lambda) \left( \YH - e_n \DI \right)\mb{T}(\lambda).
  \label{eqp4}
\end{align}

This holds because
\begin{align}
 \mb{T}(\lambda) \left( \mb{H} + \lambda e_n \DI \right)\mb{T}(\lambda) \ket{e_n} = e_n \ket{e_n},
  \label{eqp5}
\end{align}
which follows from adding the term \( \lambda e_n \DI \) to both sides of the unperturbed eigenvalue equation:
\[
  \left( \mb{H} + \lambda e_n \DI \right) \ket{e_n} = e_n (\mb{I} + \lambda \DI) \ket{e_n}.
\]

Next, consider the quantity \( \bra{e_n}\mb{T}(\lambda) \mb{H}(\lambda)\mb{T}(\lambda) \ket{e_n(\lambda)} \).  
We compute it two ways---acting from the right and from the left.

From the right: by the generalized eigenvalue equation,
\[
  \bra{e_n} \mb{H}_{\ms{QI}}(\lambda) \ket{e_n(\lambda)} 
  = e_n(\lambda) \braket{e_n | e_n(\lambda)}.
\]

From the left: using Eqs. \eqref{eqp4} and \eqref{eqp5},
\begin{align*}
  \bra{e_n}\mb{H}_{\ms{QI}}(\lambda) \ket{e_n(\lambda)} 
  &= \bra{e_n}\mb{T}(\lambda) \left( \mb{H} + \lambda e_n \DI \right)\mb{T}(\lambda) \ket{e_n(\lambda)} \\
  &\quad + \lambda \bra{e_n}\mb{T}(\lambda) \left( \YH - e_n \DI \right)\mb{T}(\lambda) \ket{e_n(\lambda)} \\
  &= e_n \braket{e_n | e_n(\lambda)} 
  + \lambda \bra{e_n}\mb{T}(\lambda) \left( \YH - e_n \DI \right)\mb{T}(\lambda) \ket{e_n(\lambda)}.
\end{align*}

Now expand \( \ket{e_n(\lambda)} = \ket{e_n} + \lambda \ket{e_n^{(1)}} + \lambda^2 \ket{e_n^{(2)}} + \cdots \) as a power series in \( \lambda \).  
Assuming orthogonality \( \braket{e_n | e_n^{(k)}} = 0 \) for \( k \geq 1 \), we have \( \braket{e_n | e_n(\lambda)} = 1 \).  
Therefore,
\[
  e_n(\lambda) - e_n 
  = \lambda \bra{e_n}\mb{T}(\lambda) \left( \YH - e_n \DI \right)\mb{T}(\lambda) \ket{e_n(\lambda)}.
\]\end{proof}

\begin{mdframed}
\begin{proposition}
  \label{cor:lambda-small}
  For \( 0 < \lambda < 1 \), suppose the eigenvalue \( e_n \) is non-degenerate, i.e., \( |e_k^{(0)} - e_n| \) is bounded away from zero for all \( k \neq n \). Then
  \begin{align*}
    e_n(\lambda) - e_n 
    &\approx \lambda^2 \left( 
      \bra{e_n}(\YH - e_n \DI)\DI \ket{e_n} 
      - \sum_{k \neq n} \tfrac{ 
        \bra{e_n} (\YH - e_n \DI) \ket{e_k^{(0)}} 
        \bra{e_k^{(0)}} \YH \ket{e_n} 
      }{e_k^{(0)} - e_n} 
    \right).
  \end{align*}
\end{proposition}
\end{mdframed}

\begin{proof}
From Theorem~\ref{thm:2.1}, we have
\[
  e_n(\lambda) - e_n 
  = \lambda \bra{e_n}\mb{T}(\lambda) \left( \YH - e_n \DI \right)\mb{T}(\lambda) \ket{e_n(\lambda)}.
\]

Up to first-order approximation, write \( \ket{e_n(\lambda)} \approx \ket{e_n} + \lambda \ket{e_n^{(1)}} \). Then
\begin{align}
  e_n(\lambda) - e_n 
  &\approx \lambda \bra{e_n}\mb{T}(\lambda) (\YH - e_n \DI)\mb{T}(\lambda) \ket{e_n} 
  + \lambda^2 \bra{e_n}\mb{T}(\lambda) (\YH - e_n \DI)\mb{T}(\lambda) \ket{e_n^{(1)}}.
  \label{eqp6}
\end{align}

We now approximate
\[
\mb{T}(\lambda) = (\mb{I} + \lambda \DI)^{-1/2} \approx \mb{I} - \tfrac{1}{2} \lambda \DI.
\]
Hence,
\begin{align*}
 \mb{T}(\lambda) (\YH - e_n \DI)\mb{T}(\lambda) 
  &\approx (\mb{I} - \tfrac{1}{2} \lambda \DI) (\YH - e_n \DI)(\mb{I} - \tfrac{1}{2} \lambda \DI) \\
  &= (\YH - e_n \DI) - \lambda (\YH - e_n \DI) \DI + \tfrac{1}{4} \lambda^2 \DI^2 \\
  &\approx (\YH - e_n \DI) - \lambda (\YH - e_n \DI) \DI,
\end{align*}
where we used the identity \( \DI \YH = \YH \DI \).

Now note that \( \bra{e_n} (\YH - e_n \DI) \ket{e_n} = 0 \), since both \( \YH \) and \( \DI \) are zero on the diagonal blocks.

Therefore, the first term of Eq.~\eqref{eqp6} becomes
\[
\lambda \bra{e_n}\mb{T}(\lambda) (\YH - e_n \DI)\mb{T}(\lambda) \ket{e_n}
\approx \lambda^2 \bra{e_n} (\YH - e_n \DI) \DI \ket{e_n},
\]
and the second term becomes
\[
\lambda^2 \bra{e_n}\mb{T}(\lambda) (\YH - e_n \DI)\mb{T}(\lambda) \ket{e_n^{(1)}}
\approx \lambda^2 \bra{e_n} (\YH - e_n \DI) \ket{e_n^{(1)}}.
\]

To compute \( \ket{e_n^{(1)}} \), note that
\[
\mb{H}_{\ms{QI}}(\lambda)
\approx \mb{H} + \lambda \YH + \lambda^2 \YH \DI + \cdots,
\]
so the first-order correction \( \ket{e_n^{(1)}} \) agrees with standard non-degenerate perturbation theory for \( \mb{H} + \lambda \YH \).
That is,
\[
\ket{e_n^{(1)}} = -\sum_{k \neq n} \tfrac{ \ket{e_k^{(0)}} \YH_{kn} }{e_k^{(0)} - e_n^{(0)}},
\qquad \text{where } \YH_{kn} = \bra{e_k^{(0)}} \YH \ket{e_n}, \text{ and } e_n = e_n^{(0)}.
\]

Then,
\begin{align*}
  \bra{e_n} (\YH - e_n \DI) \ket{e_n^{(1)}} 
  &= -\sum_{k \neq n} \tfrac{ \bra{e_n} (\YH - e_n \DI) \ket{e_k^{(0)}} \YH_{kn} }{e_k^{(0)} - e_n^{(0)}} \\
  &= -\sum_{k \neq n} \tfrac{ \bra{e_n} (\YH - e_n \DI) \ket{e_k^{(0)}} \bra{e_k^{(0)}} \YH \ket{e_n} }{e_k^{(0)} - e_n^{(0)}}.
\end{align*}

Thus, the full second-order correction is
\[
e_n(\lambda) - e_n \approx \lambda^2 \left( 
  \bra{e_n} (\YH - e_n \DI) \DI \ket{e_n} 
  - \sum_{k \neq n} \tfrac{ \bra{e_n} (\YH - e_n \DI) \ket{e_k^{(0)}} \bra{e_k^{(0)}} \YH \ket{e_n} }{e_k^{(0)} - e_n}
\right).
\]
\end{proof}

\begin{mdframed}
\begin{corollary}
  \label{cor:lambda-one}
  If \( |e_k^{(0)} - e_n| \) is bounded away from zero for all \( k \neq n \) (i.e., \( e_n \) is non-degenerate), then
  \begin{align}
    e_n(1) - e_n 
    \approx \bra{e_n}(\YH - e_n \DI) \DI \ket{e_n} 
    - \sum_{k \neq n} \tfrac{ 
        \bra{e_n} (\YH - e_n \DI) \ket{e_k^{(0)}} 
        \bra{e_k^{(0)}} \YH \ket{e_n} 
      }{e_k^{(0)} - e_n}.  
   \label {eq:corr}
  \end{align}
\end{corollary}
\end{mdframed}

 \subsubsection{Exact Bare Eigenvalues and Eigenstates and Some Combinatorial Facts}

In this section, we recall the exact eigenvalues and eigenstates of the bare subsystems \( \mb{H}_L \) and \( \mb{H}_R \)
(from Section 7 of \cite{Choi-Beyond}), and state several combinatorial identities that will be used in deriving the bounds in the next section.

For simplicity, we assume the unweighted case with \( w_r = w_l = w (=1) \), and take \( R \) to be the unique global minimum with \( n_R = 1 \). We also assume that all cliques in \( L \) have uniform size \( n_c = n_L \), and that \( m_r =: m_R > m_L:= m_l \) but \( m_R < m_L \sqrt{n_c} \), so that the bare (i.e., unperturbed) energies cross at some value \( \mt{x} = x_c \).

\begin{remark}[Notation]
In describing the bipartite graphs (\Gdis{} and \Gshare{}), 
we use \( m_l, m_r \) and \( n_l, n_r \) for the sizes of the left and right
vertex sets and cliques, respectively.
In the following analytical sections, we switch to \( m_L, m_R \) and \( n_L, n_R \),
matching the convention \( \chi_A \) for \( A \in \{L,R\} \) used for other indexed quantities
(e.g., \( E_{A_k} \), \(\ket{A_k} \)).
These pairs of symbols refer to the same quantities.
\end{remark}

\paragraph{Application of Exact Spectrum.}
We apply the exact spectrum and eigenstates of the bare subsystem for \( \mb{H}_L \), \( \mb{H}_R \) from Theorem 7.1 in \cite{Choi-Beyond}
to the case where \( n_i = n_A \) for all \( i \in A \in \{L, R\} \), and \( w = 1 \). This yields the following exact form for the bare eigenvalues and eigenstates:

\begin{mdframed}
\begin{theorem}[Bare Eigenvalues and Ground States for \( \mb{H}_L \), \( \mb{H}_R \)]
For \( A \in \{L, R\} \), the eigenvalues of \( \mb{H}_A \) are given by:
\begin{align}
  e_{A_k} = -\left( \tfrac{m_A}{2}w + (\tfrac{m_A}{2} - k) \sqrt{w^2 + n_A \mt{x}^2} \right),
  \quad k = 0, \ldots, m_A.
\end{align}
($k$ is the number of one's).

Each eigenstate is indexed by a bit string \( z = z_1 z_2 \ldots z_{m_A} \in \{0,1\}^{m_A} \), with:
\begin{align}
  \ket{E_z^{(A)}} = \bigotimes_{i=1}^{m_A} \ket{\beta_{z_i}^{(A)}},
\end{align}
where \( \ket{\beta_0^{(A)}} \) and \( \ket{\beta_1^{(A)}} \) are the eigenvectors of the local two-level system:
\begin{align*}
\begin{cases}
\ket{\beta_0^{(A)}} &= \tfrac{1}{\sqrt{1 + \gamma_A^2}} \left( \gamma_A \ket{0} +  \ket{1} \right), \\[4pt]
\ket{\beta_1^{(A)}} &= \tfrac{1}{\sqrt{1 + \gamma_A^2}} \left( \ket{0} - \gamma_A \ket{1} \right),
\end{cases}
\end{align*}
with mixing ratio
\(
\gamma_A = \tfrac{ \sqrt{n_A} \mt{x} }{ w + \sqrt{ w^2 + n_A \mt{x}^2 } }.
\)

\paragraph{Ground State Energies and Vectors.}
In particular, the ground states of \( \mb{H}_L \) and \( \mb{H}_R \) correspond to the all-zero string \( z = 0\ldots 0 \), and take the form:
\begin{align*}
  \begin{cases}
  \ket{R_0} &= \bigotimes_{i=1}^{m_R} \ket{\beta_0^{(R)}}, \\[4pt]
  \ket{L_0} &= \bigotimes_{i=1}^{m_L} \ket{\beta_0^{(L)}}.
  \end{cases}
\end{align*}

The corresponding ground state energies are:
\begin{align*}
\begin{cases}
e_{L_0} &= -\tfrac{m_L}{2} \left( w + \sqrt{w^2 + n_L \mt{x}^2} \right), \\
e_{R_0} &= -\tfrac{m_R}{2} \left( w + \sqrt{w^2 + n_R \mt{x}^2} \right).
\end{cases}
\end{align*}
\end{theorem}
\end{mdframed}

We define the overlap of an eigenstate \( \ket{A_k} \), where \( k \) denotes the number of 1's in the bit string indexing the state, with the all-zero computational basis state as
\[
\co(A_k) \mdef \braket{A_k \mid 0_A},
\]
where \( \ket{0_A} \equiv \ket{0}^{\otimes m_A} \) is the zero state on subsystem \( A \in \{L, R\} \).

The following combinatorial identities will be used in the gap analysis:

\begin{proposition}
Let \( A \in \{L, R\} \) and let \( \ket{A_k} \) denote any eigenstate with Hamming weight \( k \). Then the overlap of \( \ket{A_k} \) with the all-zero computational state is:
\begin{align}
    \co(A_k) = \left( \tfrac{1}{\sqrt{1 + \gamma_A^2}} \right)^{m_A}  \gamma_A^{m_A-k}
\end{align}
In particular, for the ground state \( \ket{A_0} \), corresponding to \( k = 0 \), we have:
\begin{align}
    \co(A_0) = \left( \tfrac{\gamma_A}{\sqrt{1 + \gamma_A^2}} \right)^{m_A}
    \label{eq:co}
\end{align}
\end{proposition}

One can verify the following non-obvious identity:
\begin{lemma}
\label{lemma:surprising-id}
Let \( \gamma_A = \tfrac{\sqrt{n_A} \mt{x}}{1 + \sqrt{1 + n_A \mt{x}^2}} \), we have
\(
\tfrac{1 - \gamma_A^2}{1 + \gamma_A^2} = \tfrac{1}{\sqrt{1 + n_A \mt{x}^2}}.
\)
\end{lemma}

Together with the two well-known binomial summation identities:
\begin{align}
  \sum_{k=0}^{m} \binom{m}{k} a^k &= (1+a)^m, \quad
  \sum_{k=0}^{m} \binom{m}{k} k a^k = \tfrac{a}{1+a}\, m (1+a)^m,
  \label{eq:binom-sum}
\end{align}
we have the folowing perhaps surprising identity.
\begin{lemma}
\label{prop:energy-weighted-sum}
Let \( A \in \{L, R\} \), and let \( e_{A_k} \) denote the energy of the eigenstate \( \ket{A_k} \) with Hamming weight \( k \). Then
\begin{align}
  \sum_{k=0}^{m_A} e_{A_k} \left( \braket{A_k \mid 0_A} \right)^2 = 0.
\end{align}
\end{lemma}

\begin{proof}
We begin by expressing the energy as
\(
e_{A_k}=  e_{A_0} + k \cdot \sqrt{1 + n_A \mt{x}^2}.
\)
The squared overlap of the eigenstate \( \ket{A_k} \) with the all-zero computational basis state is
\[
\left( \braket{A_k \mid 0_A} \right)^2 = \left( \tfrac{1}{1 + \gamma_A^2} \right)^{m_A} (\gamma_A^2)^{m_A-k}.
\]

There are \( \binom{m_A}{k} \) such states at each Hamming weight \( k \), so the full sum becomes:
\begin{align*}
S &:= \sum_{k=0}^{m_A} \binom{m_A}{k} e_{A_k} \left( \braket{A_k \mid 0_A} \right)^2 \\
&= \left( \tfrac{1}{1 + \gamma_A^2} \right)^{m_A}
\sum_{k=0}^{m_A} \binom{m_A}{k} \left( e_{A_0} + k \cdot \sqrt{1 + n_A \mt{x}^2} \right) (\gamma_A^2)^{m_A-k}.
\end{align*}

By the binomial summation identities in Eq.~\eqref{eq:binom-sum}, we have
\begin{align*}
S &= \left( \tfrac{1}{1 + \gamma_A^2} \right)^{m_A} \cdot
\left[
(1 + \gamma_A^2)^{m_A} \left( e_{A_0} + \tfrac{1}{1 + \gamma_A^2} m_A \cdot \sqrt{1 + n_A \mt{x}^2} \right)
\right] \\
&= e_{A_0} + \tfrac{1}{1 + \gamma_A^2} \cdot m_A \cdot \sqrt{1 + n_A \mt{x}^2}\\
&= -\left( \tfrac{m_A}{2} + \tfrac{m_A}{2} \sqrt{1+ n_A \mt{x}^2} \right) +  \tfrac{1}{1 + \gamma_A^2} \cdot m_A \cdot \sqrt{1 + n_A \mt{x}^2}\\
&= -\tfrac{m_A}{2} - \tfrac{m_A}{2} \sqrt{1+ n_A \mt{x}^2} \tfrac{\gamma_A^2-1}{1 + \gamma_A^2} \\
&= 0 \qquad \text{(by Lemma~\ref{lemma:surprising-id})}.
\end{align*}
\end{proof}


\subsubsection{Anti-Crossing Structure: Energy Approximation and Gap Bound}
\label{sec:ac-gap}

\subsubsection*{Ground State Energy Approximation}
Let \( e_{A_0} \) denote the bare ground state energy, and let \( E_0 = e_{A_0}(1) \) be the true ground state energy of \( \Hcore \).
We apply the correction formula in Eq.~\eqref{eq:corr} from
Corollary~\ref{cor:lambda-one} to the ground state energy before and after the anti-crossing to obtain an approximation
of the difference \( E_0 - e_{A_0} \).
Assume the bare energy ordering satisfies \( e_{L_0} < e_{R_0} \) before the anti-crossing, and reverses to \( e_{R_0} < e_{L_0} \) afterward.  
Suppose further that the energy separation satisfies \( |e_{L_0} - e_{R_0}| > 1 \) (i.e. away from the anti-crossing), so that a first-order correction suffices.  
Then the corrected ground state energies can be written explicitly as weighted sums over overlaps and energy gaps between the \( L \) and \( R \) blocks, as given in Proposition~\ref{prop:explicit-correction}.

\begin{proposition}
  \label{prop:explicit-correction}
  Suppose that before the anti-crossing, \( e_{L_0} < e_{R_0} \), and that this ordering reverses afterward.  
  Assume further that \( |e_{L_0} - e_{R_0}| > 1 \), so that a first-order correction suffices.
  Let \( E_0 \) denote the true ground state energy of \( \Hcore \).

  \begin{itemize}
    \item Before the anti-crossing:
    \[
      E_0 - e_{L_0} \approx 
      -2 e_{L_0} \sum_j \tfrac{e_{R_j}}{e_{R_j} - e_{L_0}} 
      \left( \braket{\bar{L}_0 | \bar{R}_j} \right)^2.
    \]
    
    \item After the anti-crossing:
    \[
      E_0 - e_{R_0} \approx 
      -2 e_{R_0} \sum_j \tfrac{e_{L_j}}{e_{L_j} - e_{R_0}} 
      \left( \braket{\bar{R}_0 | \bar{L}_j} \right)^2.
    \]
  \end{itemize}
\end{proposition}

\begin{proof}
We apply the correction formula in Eq.~\eqref{eq:corr} from
Corollary~\ref{cor:lambda-one} to the ground state energy before and after the anti-crossing:
\begin{itemize}
    \item Before the anti-crossing:
    \begin{align}
      E_0 - e_{L_0} \approx 
      \sum_j \bra{e_{L_0}} ( \YH - e_{L_0} \DI ) \ket{e_{R_j}} 
      \left( \bra{e_{R_j}} \DI \ket{e_{L_0}} - \tfrac{ \bra{e_{R_j}} \YH \ket{e_{L_0}} }{ e_{R_j} - e_{L_0} } \right)
    \end{align}

    \item After the anti-crossing:
    \begin{align}
      E_0 - e_{R_0} \approx 
      \sum_j \bra{e_{R_0}} ( \YH - e_{R_0} \DI ) \ket{e_{L_j}} 
      \left( \bra{e_{L_j}} \DI \ket{e_{R_0}} - \tfrac{ \bra{e_{L_j}} \YH \ket{e_{R_0}} }{ e_{L_j} - e_{R_0} } \right)
    \end{align}
\end{itemize}
We now use the identities
\[
  \bra{e_{L_i}} \DI \ket{e_{R_j}} = \braket{\bar{L}_i | \bar{R}_j}, \quad 
  \bra{e_{L_i}} \YH \ket{e_{R_j}} = (e_{L_i} + e_{R_j}) \braket{\bar{L}_i | \bar{R}_j},
\]
where \( \ket{\bar{L}_i} \) and \( \ket{\bar{R}_j} \) denote the normalized bare eigenstates in the \( L \) and \( R \) subsystems, respectively. Substituting into the above yields the result.
\end{proof}

We justify that the first-order correction to \( e_{L_0} \) is small by analyzing the weighted sum:
\[
\sum_j \tfrac{e_{R_j}}{e_{R_j} - e_{L_0}} 
\left( \braket{ \bar{R}_j \mid 0_R } \right)^2.
\]
By Lemma~\ref{prop:energy-weighted-sum}, the unweighted sum is exactly zero:
\(
\sum_j e_{R_j} \left( \braket{ \bar{R}_j \mid 0_R } \right)^2 = 0
\).
The weighting factor \( \tfrac{1}{e_{R_j} - e_{L_0}} \) is smooth and strictly positive across the spectrum, since \( e_{R_j} > e_{L_0} \) for all \( j \) before the anti-crossing.
Thus, the weighted sum remains bounded in magnitude and does not grow with system size.
Moreover, the prefactor \( \left( \braket{L_0 \mid 0_L} \right)^2 = \co(L_0)^2 \) is exponentially small in \( m_L \), so the total correction is exponentially suppressed:
\(
E_0 - e_{L_0} \approx -2 e_{L_0} \cdot \co(L_0)^2 \cdot \left[ \text{bounded sum} \right].
\)
Hence, the correction is negligible, and we conclude that
\(
E_0 \approx e_{L_0}.
\)
This justifies the approximation \( E_0 \approx e_{L_0} \) before the anti-crossing,  
and similarly \( E_0 \approx e_{R_0} \) after the anti-crossing,  
completing the perturbative approximation scheme.
See Figure~\ref{fig:energy-Hcore} for numerical confirmation.
\begin{figure}[!htbp]
  \centering
  $$
  \begin{array}{cc}
  \includegraphics[width=0.45\textwidth]{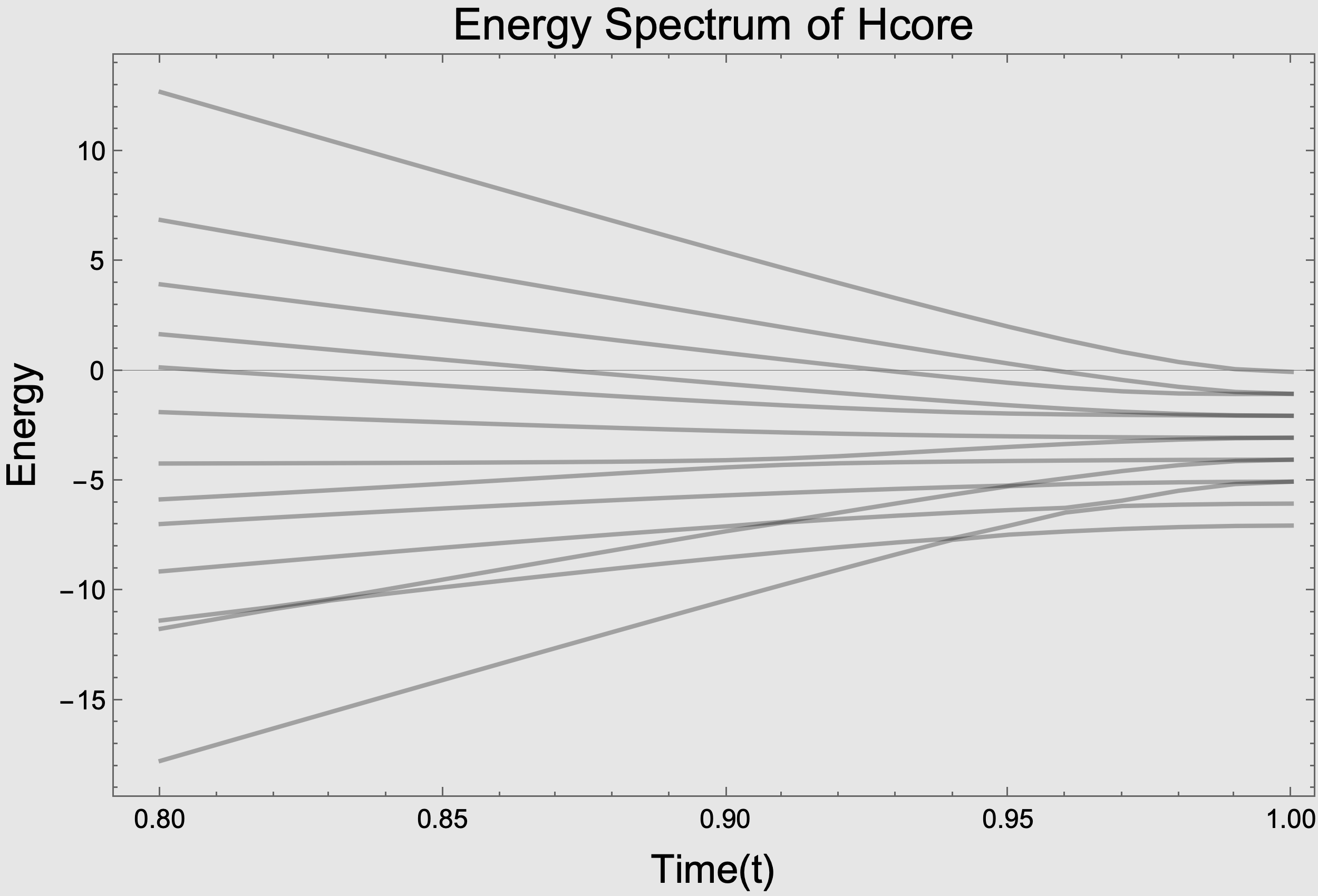} &
  \includegraphics[width=0.45\textwidth]{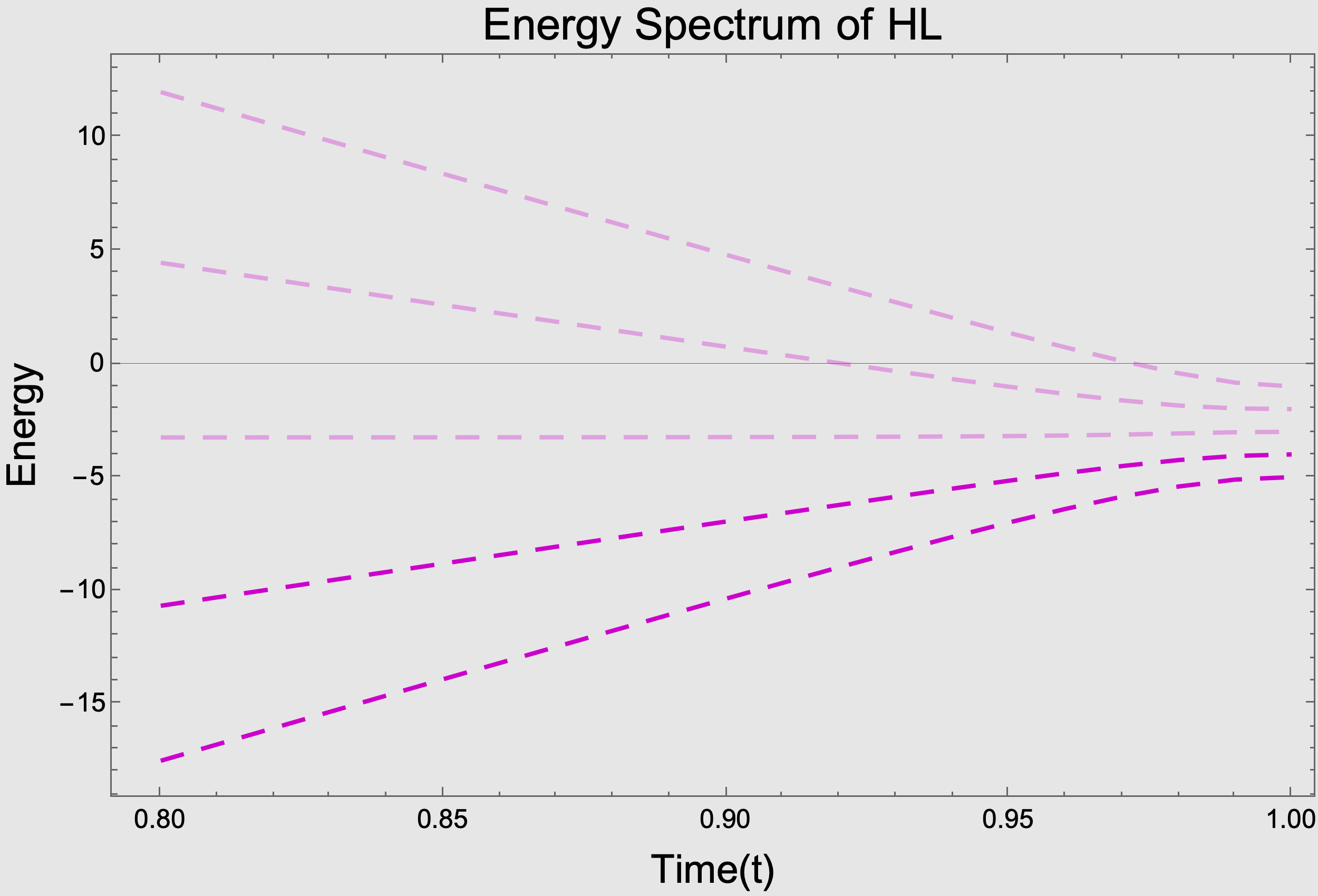}\\
  \includegraphics[width=0.45\textwidth]{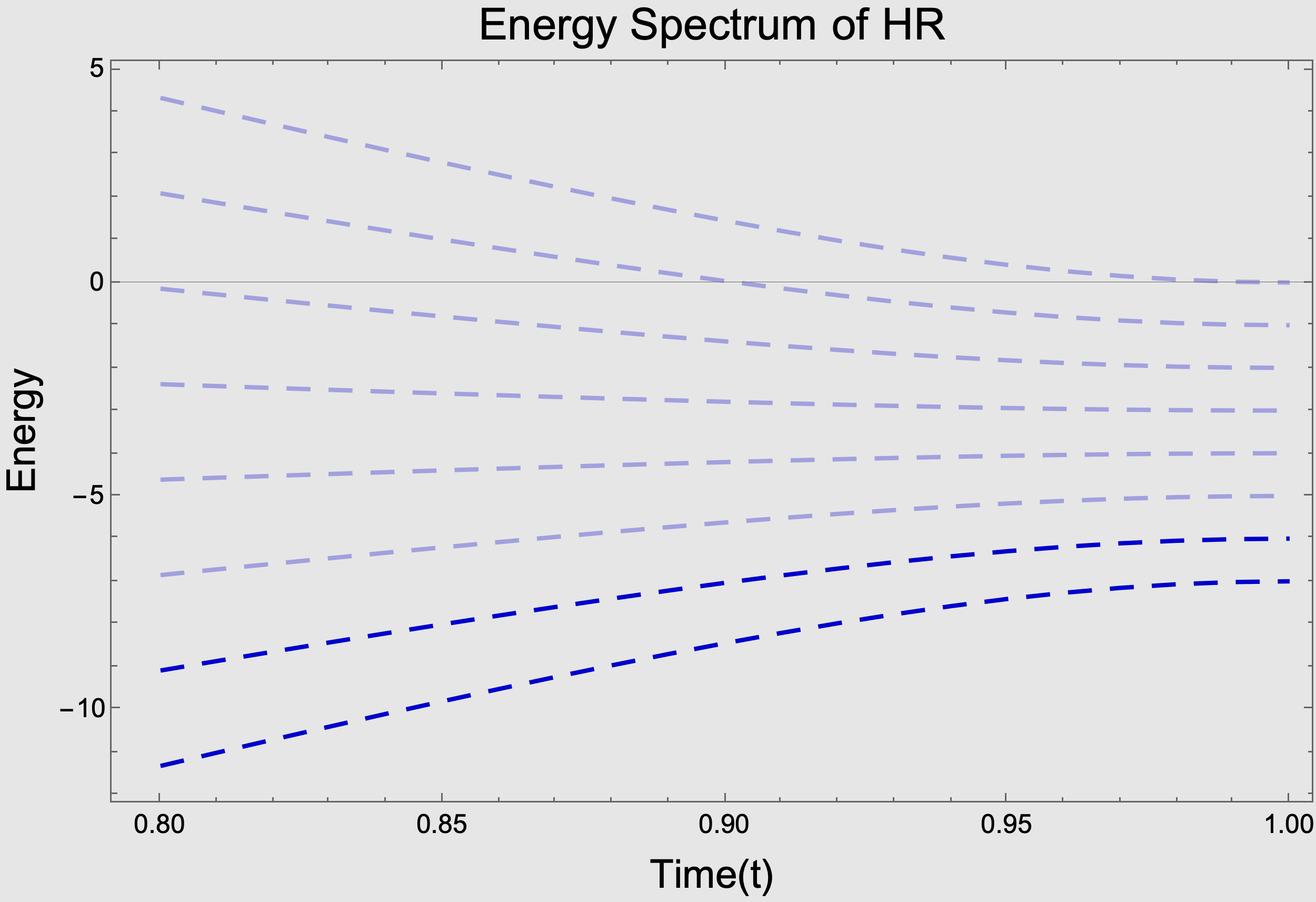} &
  \includegraphics[width=0.45\textwidth]{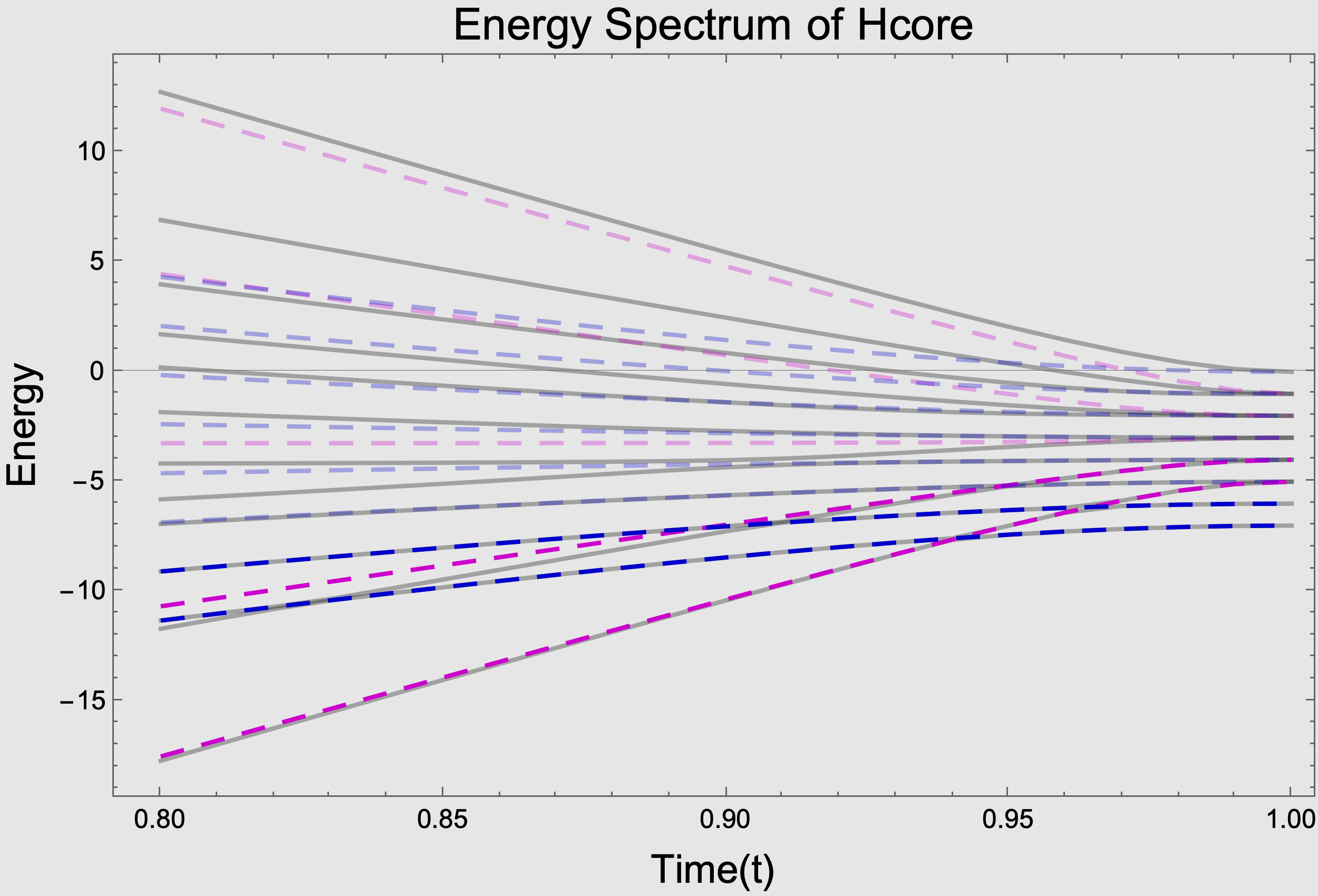}
  \end{array}
  $$

  \caption{
    Comparison of the exact energy spectrum of the effective Hamiltonian \( \Hcore \)  
    with the bare energy spectra of \( \mb{H}_L \) (magenta dashed) and \( \mb{H}_R \) (blue dashed).  
    The top-left panel shows the true spectrum of \( \Hcore \); the top-right and bottom-left show the bare spectra of  
    \( \mb{H}_L \) and \( \mb{H}_R \), respectively.  
    The bottom-right panel overlays the bare and true spectra.
    As seen, the true energy (solid gray) of \( \Hcore \) closely follows the bare energy (dashed),  
    and the bare-level crossing is replaced by an anti-crossing.
  }

  \label{fig:energy-Hcore}
\end{figure}
Furthermore, the actual anti-crossing point is well-approximated by the crossing point of the bare energies,  
as illustrated in Figure~\ref{fig:anticrossing-combined}.
\begin{figure}[!htbp]
  \centering
  $$
  \begin{array}{cc}
    \includegraphics[width=0.48\textwidth]{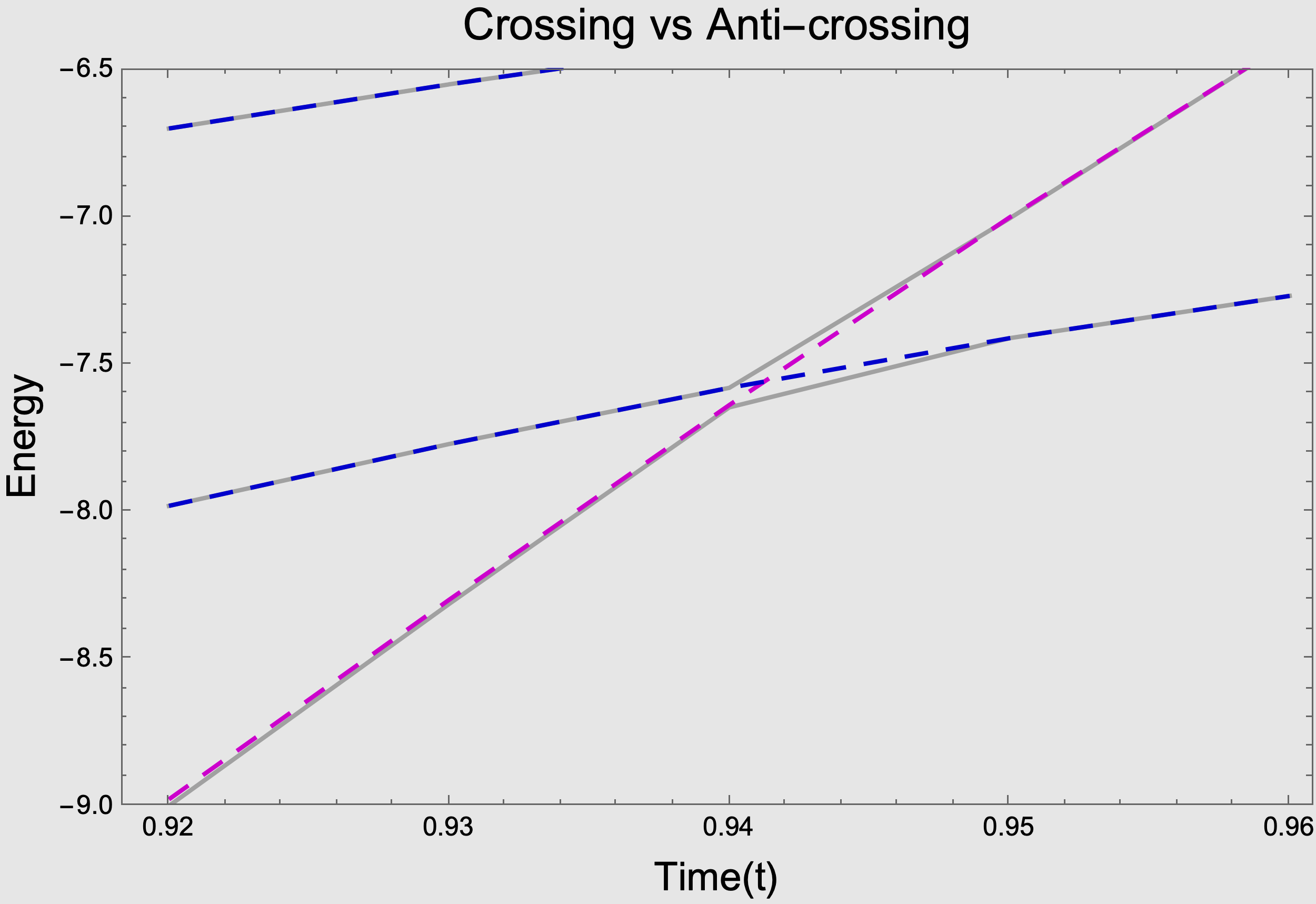} &
    \includegraphics[width=0.48\textwidth]{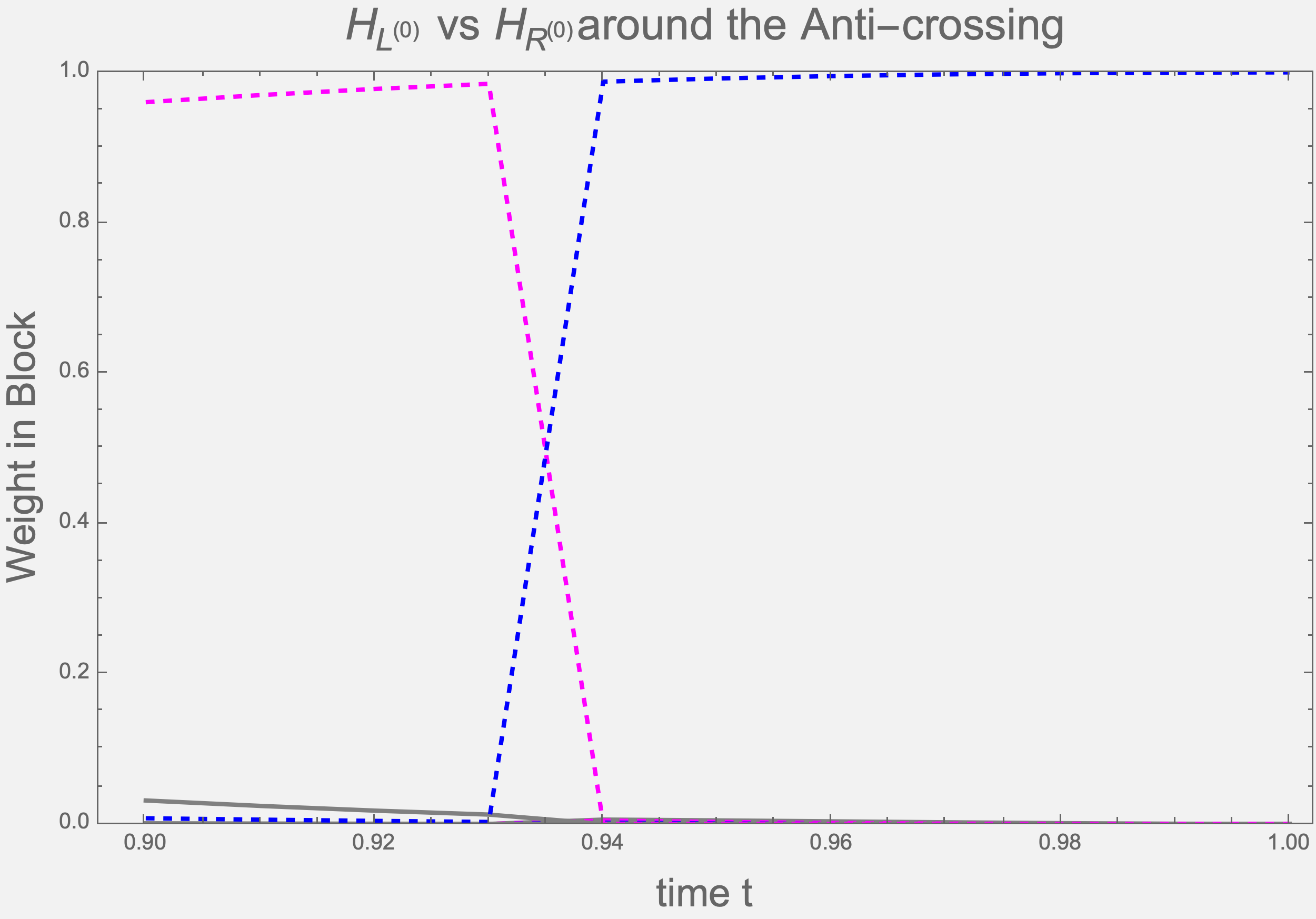}
  \end{array}
  $$
  \caption{
  Anti-crossing illustration.  
  Left: Energy spectrum comparison between the bare Hamiltonians  
  \( \mb{H}_L^{(0)} \) and \( \mb{H}_R^{(0)} \) (dashed magenta and blue),  
  and the coupled Hamiltonian \( \Hcore \) (solid gray).  
  The anti-crossing occurs near the point where the bare energies cross.  
  Right: Ground state projection onto the same two blocks.  
  The ground state is concentrated on \( \mb{H}_L^{(0)} \)  (magenta) before the anti-crossing,  
  and \( \mb{H}_R^{(0)} \) (blue) after the anti-crossing.
}

  \label{fig:anticrossing-combined}
\end{figure}

\subsubsection*{Effective \( 2 \times 2 \) Hamiltonian and Gap Bound}

We now approximate the size of the anti-crossing gap between the degenerate bare states \( \ket{e_{L_0}} \) and \( \ket{e_{R_0}} \), by reducing the full generalized eigenvalue problem to an effective \( 2 \times 2 \) Hamiltonian on the subspace they span.
Combining the structural reduction steps established above, we have the chain of approximations:
\setlength{\fboxsep}{8pt} 
\begin{equation}
  \label{eq:gap-chain}
  \fbox{$
    \Delta(\mb{H}) 
    \ \le\ \Delta(\mb{H}^{\ms{low}}) 
    \ \approx\ \Delta(\mb{H}_{\mc{C}}) 
    \ \approx\ \Delta(\Hcore) 
    \ \approx\ \Delta(\Hefftwo)
  $}
\end{equation}
which shows that it suffices to approximate the gap of the final \( 2\times 2 \) Hamiltonian \(\Hefftwo\).
We now derive \(\Hefftwo\) explicitly, starting from the projected generalized eigenvalue problem.
This approach is similar to~\cite{AndradeFreire2003}. 
See also \cite{Kerman-Tunnel}.

We construct a two-level approximation by projecting the generalized eigenvalue problem onto the subspace
\[
P = \text{span} \{ \ket{\bar{L}_0}, \ket{\bar{R}_0} \},
\]
where \( \ket{\bar{L}_0} \) and \( \ket{\bar{R}_0} \) are the padded bare eigenstates of \( \mb{H}_L^{(0)} \) and \( \mb{H}_R^{(0)} \), respectively.  
Projecting the generalized eigenvalue problem onto \(P\) yields the intermediate effective Hamiltonian \(\Heffzero\) in generalized form; converting it to standard form by left-multiplying with \(\mb{S}_{PP}^{-1}\) produces the final \(2\times 2\) Hamiltonian \(\Hefftwo\).

\begin{lemma}
\label{claim:effective-hamiltonian}
The effective Hamiltonian for the subspace \( P \) is given by
\[
\Heffzero 
= \mb{H}_{PP} 
+ (e_0 \mb{S}_{PQ} - \mb{H}_{PQ})(e_0 \mb{S}_{QQ} - \mb{H}_{QQ})^{-1}(e_0 \mb{S}_{QP} - \mb{H}_{QP}),
\]
where \( e_0 = e_{L_0}(0) = e_{R_0}(0) \), $Q=P^{\perp}$.
The eigenvalues of the pair \( (\Heffzero, \mb{S}_{PP}) \) approximate the ground and first excited energies near the anti-crossing.
\end{lemma}

\begin{proof}
The perturbed generalized eigenvalue equation is:
\begin{align}
  (\mb{H} + \YH) \ket{\Phi} = \lambda (\mb{I} + \DI) \ket{\Phi}.
\end{align}

Write the eigenvector as 
\[
\ket{\Phi} = \sum_{p \in P} c_p \ket{p} + \sum_{q \in Q} c_q \ket{q},
\]
and take inner products to obtain the generalized eigenvalue problem in matrix form:
\begin{align}
  \begin{bmatrix}
    \mb{H}_{PP} & \mb{H}_{PQ} \\
    \mb{H}_{QP} & \mb{H}_{QQ}
  \end{bmatrix}
  \begin{bmatrix}
    c_P \\
    c_Q
  \end{bmatrix}
  = \lambda
  \begin{bmatrix}
    \mb{S}_{PP} & \mb{S}_{PQ} \\
    \mb{S}_{QP} & \mb{S}_{QQ}
  \end{bmatrix}
  \begin{bmatrix}
    c_P \\
    c_Q
  \end{bmatrix},
  \label{eq:eff1}
\end{align}
where \( \lambda = e_0 + \Delta_\lambda \).

From Eq.~\eqref{eq:eff1}, we obtain:
\[
\left\{
\begin{array}{l}
(\mb{H}_{PP} - \lambda \mb{S}_{PP}) c_P = (\lambda \mb{S}_{PQ} - \mb{H}_{PQ}) c_Q, \\
(\mb{H}_{QP} - \lambda \mb{S}_{QP}) c_P = (\lambda \mb{S}_{QQ} - \mb{H}_{QQ}) c_Q.
\end{array}
\right.
\]

We solve the second equation for \( c_Q \):
\(
  c_Q = (\lambda \mb{S}_{QQ} - \mb{H}_{QQ})^{-1} (\mb{H}_{QP} - \lambda \mb{S}_{QP}) c_P.
\)

Substitute into the first equation:
\begin{align*}
(\mb{H}_{PP} - \lambda \mb{S}_{PP}) c_P 
= (\lambda \mb{S}_{PQ} - \mb{H}_{PQ})(\lambda \mb{S}_{QQ} - \mb{H}_{QQ})^{-1} (\mb{H}_{QP} - \lambda \mb{S}_{QP}) c_P.
\end{align*}
\[
\Longrightarrow 
\left[
\mb{H}_{PP} + (\lambda \mb{S}_{PQ} - \mb{H}_{PQ})(\lambda \mb{S}_{QQ} - \mb{H}_{QQ})^{-1} (\lambda \mb{S}_{QP} - \mb{H}_{QP})
\right] c_P = \lambda \mb{S}_{PP} c_P.
\]

Since \( \lambda = e_0 + \Delta_\lambda \) and \( |\Delta_\lambda| \ll |e_0| \),  
we may approximate the terms \( \lambda \mb{S}_{PQ} - \mb{H}_{PQ} \),  
\( \lambda \mb{S}_{QQ} - \mb{H}_{QQ} \), and \( \lambda \mb{S}_{QP} - \mb{H}_{QP} \)  
by their counterparts evaluated at \( \lambda = e_0 \), provided that the overlaps 
\( \braket{\bar{L}_0 | \bar{R}_j} \) are small.

Concretely, the matrix element
\(
\bra{\bar{L}_0} (\lambda \mb{S}_{PQ} - \mb{H}_{PQ}) \ket{\bar{R}_j} 
= \Delta_\lambda \braket{\bar{L}_0 | \bar{R}_j} - e_{R_j} \braket{\bar{L}_0 | \bar{R}_j},
\)
differs from the \( e_0 \)-version only by a term of order \( \Delta_\lambda \cdot \braket{\bar{L}_0 | \bar{R}_j} \),  
which is negligible under the assumption that both \( \Delta_\lambda \ll |e_0| \) and \( \braket{\bar{L}_0 | \bar{R}_j} \ll 1 \).  
Thus, we may replace \( \lambda \) by \( e_0 \) in the effective Hamiltonian to leading order.

Substituting \( \lambda \approx e_0 \), we obtain the effective Hamiltonian:
\[
\Heffzero c_P \approx \lambda \mb{S}_{PP} c_P,
\]
where
\[
\Heffzero 
= \mb{H}_{PP} 
+ (e_0 \mb{S}_{PQ} - \mb{H}_{PQ})(e_0 \mb{S}_{QQ} - \mb{H}_{QQ})^{-1}(e_0 \mb{S}_{QP} - \mb{H}_{QP}).
\]
\end{proof}

To understand the leading-order contribution to the anti-crossing gap, we examine the structure of the effective Hamiltonian in the \( P \)-subspace.
We expand the second-order approximation using
\[
\mb{D} \mdef e_0 \mb{I}_{QQ} - \mb{E}_{QQ}, \quad 
\mb{V} \mdef \YH_{QQ} - e_0 \DI_{QQ},
\]
and approximate the inverse:
\(
(e_0 \mb{S}_{QQ} - \mb{H}_{QQ})^{-1} 
\approx \mb{D}^{-1} + \mb{D}^{-1} \mb{V} \mb{D}^{-1}.
\)
Then we have
\(
\Heffzero 
\approx \mb{H}_{PP} 
+ \mb{W}_1 + \mb{W}_2,
\)
where
\[
\mb{W}_1 \mdef (e_0 \mb{S}_{PQ} - \mb{H}_{PQ}) \mb{D}^{-1} (e_0 \mb{S}_{QP} - \mb{H}_{QP}), \quad
\mb{W}_2 \mdef (e_0 \mb{S}_{PQ} - \mb{H}_{PQ}) \mb{D}^{-1} \mb{V} \mb{D}^{-1} (e_0 \mb{S}_{QP} - \mb{H}_{QP}).
\]
The leading term
\[
\mb{H}_{PP} = 
\begin{bmatrix}
e_{L_0} & (e_{L_0} + e_{R_0}) g_0 \\
(e_{L_0} + e_{R_0}) g_0 & e_{R_0}
\end{bmatrix}
= e_0
\begin{bmatrix}
1 & 2 g_0 \\
2 g_0 & 1
\end{bmatrix},
\]
where \( e_0 = \tfrac{1}{2}(e_{L_0} + e_{R_0}) \) and \( g_0 = \braket{\bar{L}_0 | \bar{R}_0} \).

The first-order correction \( \mb{W}_1 = \begin{bmatrix} a_1 & 0 \\ 0 & a_1 \end{bmatrix} \) is diagonal, 
while the second-order correction \( \mb{W}_2 = \begin{bmatrix} 0 & c_2 \\ c_2 & 0 \end{bmatrix} \) is off-diagonal, with \( c_2 \ll e_0 g_0 \).
Combining all terms yields:
\[
\Heffzero \approx
\begin{bmatrix}
e_0 + a_1 & 2 e_0 g_0 + c_2 \\
2 e_0 g_0 + c_2 & e_0 + a_1
\end{bmatrix},
\]
which shows that the off-diagonal coupling is governed primarily by the overlap \( g_0 \).

Finally, we convert the generalized eigenvalue problem to standard form using
\[
\Hefftwo
= \mb{S}_{PP}^{-1} \Heffzero, \quad \text{where} \quad
\mb{S}_{PP} = 
\begin{bmatrix}
1 & g_0 \\
g_0 & 1
\end{bmatrix},
\]
which gives 
\[
\Hefftwo = \tfrac{1}{1 - g_0^2}
\begin{bmatrix}
e_0 + a_1 - g_0 (2 e_0 g_0 + c_2) & 2 e_0 g_0 + c_2  - g_0(e_0 + a_1) \\
2 e_0 g_0 + c_2  - g_0(e_0 + a_1) & e_0 + a_1 - g_0 (2 e_0 g_0 + c_2)
\end{bmatrix}.
\]
We thus have the gap approximation based on the assumptions that \( |a_1| \ll |e_0| \) and \( |c_2| \ll |e_0 g_0| \):
\begin{corollary}[Gap Approximation]
\label{cor:gap-bound}
Under the effective two-level approximation,
\[
\Delta(\Hefftwo) \approx 2 |e_0| g_0,
\]
where \( e_0 = e_{L_0} = e_{R_0} \) is the degenerate bare energy, and \( g_0 = \braket{\bar{L}_0 | \bar{R}_0} \) is the overlap between two bare ground states.
\end{corollary}

\begin{proposition}
The overlap between the two bare ground states is given by
\begin{align*}
  g_0 := \braket{\bar{L}_0 | \bar{R}_0}
  = \co(L_0) \cdot \co(R_0)
  = \left( \sqrt{\tfrac{\gamma_L^2}{1 + \gamma_L^2}} \right)^{m_L}
  \left( \sqrt{\tfrac{\gamma_R^2}{1 + \gamma_R^2}} \right)^{m_R},
\end{align*}
where the amplitude ratios 
\(
  \gamma_A = \tfrac{ \sqrt{n_A} \mt{x}_c }{ w + \sqrt{w^2 + n_A \mt{x}_c^2} } \in (0,1),  
\) for $A \in \{L, R\}$,
and \( \mt{x}_c \) is the location of the anti-crossing, defined implicitly by \( e_{L_0}(\mt{x}_c) = e_{R_0}(\mt{x}_c) \).
\end{proposition}

Since the runtime scales inversely with the square of the minimum gap, by the chain of approximations in Eq.~\eqref{eq:gap-chain},  we have
\[
T \gtrsim \tfrac{1}{g_0^2} 
= \left( 1+\tfrac{1}{\gamma_L^2} \right)^{m_L}
\left( 1+\tfrac{1}{\gamma_R^2} \right)^{m_R}
> 2^{m_L + m_R}. 
\]

This completes the reduction from the full Hamiltonian \(\mb{H}\) to the effective \(2\times 2\) model, providing the exponential runtime bound.

 \section{Conclusion}
\label{sec:end}

This work presents an analytical framework for identifying and quantifying tunneling-induced anti-crossing in stoquastic transverse-field quantum annealing (TFQA), based on a structured class of Maximum Independent Set instances. Our analysis reformulates the effective Hamiltonian in a non-orthogonal basis and derives an exponentially small gap using a generalized eigenvalue approach that remains valid beyond the small transverse field regime.
More broadly, our findings reveal the structural origin of tunneling-induced anti-crossings. Once the system localizes near critical local minima---as captured in the reduced form \( \Hcore \)---a tunneling transition to the global minimum becomes unavoidable, and the resulting gap is exponentially small. This same localization and tunneling mechanism may also arise in less structured graphs.

It has been claimed that non-stoquasticity may not be essential for achieving quantum speedup~\cite{Non-stoquastic,de-sign,CL2020}. Indeed, simply introducing a non-stoquastic driver is not sufficient to bypass the tunneling bottleneck. For example, if $\Jzz$ is set to be too large, the evolution 
will fail the structural steering and encounter a tunneling-induced anti-crossing instead (as illustrated by an example in 
Figure~16 of~\cite{Choi-Beyond}).

In contrast to this work, a recent study~\cite{UnstructuredAQO2024}, which may be of theoretical interest on its own, investigated unstructured adiabatic quantum optimization, aiming to recover Grover-like quadratic speedup~\cite{Grover1996,Roland2002}. However, as we note in the main paper (Section~3.1 in ~\cite{Choi-Beyond}), a simple classical algorithm~\cite{Tarjan1977} can already solve the MIS problem in time \(O(2^{n/3}) \), outperforming the Grover-like \( O(2^{n/2}) \) speedup.

This highlights the importance of algorithmic design and analysis, even for adiabatic quantum algorithms in order to achieve quantum advantage. 
We note that the tunneling-induced anti-crossings analyzed here play a constructive role in the \DDD{} algorithm:  
a polynomial-time TFQA evolution through such an anti-crossing identifies configurations in the set of the degenerate local minima,  
which are then used as seeds for constructing the non-stoquastic $\XX$-driver graph.

We hope that the structural insights developed here may guide future work in rigorously identifying similar bottlenecks in broader classes of problem instances, and in clarifying the fundamental limitations of stoquastic quantum annealing beyond heuristic or empirical claims.

\section*{Acknowledgment}
This work was written by the author with the help of ChatGPT (OpenAI),
which assisted in refining the presentation and in expressing the intended ideas with greater clarity and precision.  
The author thanks Jamie Kerman for introducing her to the angular-momentum 
basis, for discussions, and for the idea of deriving a perturbative bound on 
the anti-crossing gap by constructing an effective Hamiltonian in the 
non-orthogonal basis, as in ~\cite{AndradeFreire2003}.
We recognize that this manuscript may not cite all relevant literature.
If you believe your work should be included in a future version, please do not hesitate to contact the author with appropriate references.
The author acknowledges support from the Defense Advanced Research Projects Agency under Air Force Contract No. FA8702-15-D-0001. Any opinions, findings and conclusions or recommendations expressed in this material are those of the authors and do not necessarily reflect the views of the Defense Advanced Research Projects Agency.

\appendix

\subsection*{Appendix A: Font Conventions for Notation}
We adopt the following conventions throughout:  
\begin{itemize}
\item Hilbert space / Subspace / Basis: calligraphic, e.g.\ \(\mc{V}, \mc{B}\).  
  \item Hamiltonian / Matrix: blackboard bold, e.g.\ \(\mb{H}, \mb{B}\).  
  \item Time-dependent quantity: typewriter, e.g.\ \(\mt{x} := \mt{x}(t), \mt{jxx}:=\mt{jxx}(t)\).  
  \item Named object / Abbreviation: capital typewriter, e.g.\ \LM{}, \GM{}, \MLIS{}.  
\end{itemize}

\end{document}